\documentclass[runningheads]{llncs}

\usepackage{graphicx} 
\usepackage[firstpage]{draftwatermark} 

\SetWatermarkAngle{0}
\SetWatermarkText{\raisebox{11.5cm}{
\hspace{10.5cm}
\href{https://zenodo.org/records/15220343}
{\includegraphics{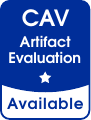}}
}}

\usepackage{float}

\usepackage[T1]{fontenc}
\usepackage[normalem]{ulem}
\usepackage[inline]{enumitem}
\usepackage[]{outlines}
\usepackage{enumitem}
\usepackage[dvipsnames]{xcolor}
\usepackage{amsmath,amssymb,comment}
\usepackage{bussproofs}
\usepackage{listings}
\usepackage{lipsum}
\usepackage{stmaryrd}
\usepackage[linesnumbered,noend]{algorithm2e} 
\SetKwInput{KwIn}{In}
\SetKwInput{KwOut}{Out}
\SetArgSty{upshape}
\SetKw{KwGoTo}{go to}
\SetKwProg{Fn}{def}{\string:}{}

\usepackage[bookmarksnumbered,unicode]{hyperref}
\hypersetup{                    
  colorlinks,
  pdfborder={0 0 0},
  pdfmenubar=false,
  pdfpagemode=UseNone,
  pdfnonfullscreenpagemode=UseNone,
  bookmarksopen=true,
  breaklinks=true,
  linkcolor={green!50!black},
  citecolor={red!50!black},
  urlcolor={blue!50!black}
}
\usepackage[hide]{comments}
\Commenter{LP}
\Commenter{AO}
\Commenter{CB}
\Commenter{ID}
\Commenter{SKP}
\Commenter{SOR}
\Commenter{AB}

\usepackage[square,comma,numbers,sort&compress]{natbib}

\usepackage{xspace}
\usepackage{tikz}
\usetikzlibrary{arrows.meta,decorations.pathreplacing,scopes,calc}
\tikzset{>=Latex[]}
\usepackage{booktabs}
\usepackage{caption}
\usepackage{subcaption}

\let\llncssubparagraph\subparagraph
\let\subparagraph\paragraph
\usepackage[compact]{titlesec}
\let\subparagraph\llncssubparagraph

\newif\ifshortversion
\shortversionfalse

\newif\ifanonymous
\anonymousfalse

\newcommand{\xspacemm}{\ifmmode\else\xspace\fi}
\newcommand{\mathnoun}[1]{\ensuremath{#1}\xspacemm}
\newcommand{\newmathnoun}[2]{\newcommand{#1}{\mathnoun{#2}}}
\newcommand{\renewmathnoun}[2]{\renewcommand{#1}{\mathnoun{#2}}}
\newcommand{\newnoun}[2]{\newcommand{#1}{#2\xspace}}

\renewcommand{\vec}{\mathbf}
\newmathnoun{\FF}{\mathbb{F}}
\newmathnoun{\ZZ}{\mathbb{Z}}
\newmathnoun{\ZZNN}{\mathbb{Z}_{\geq0}}
\newmathnoun{\PosZZ}{\mathbb{Z}_{\infty}^{+}}
\newmathnoun{\QQ}{\mathbb{Q}}
\newmathnoun{\NN}{\mathbb{N}}
\newmathnoun{\C}{\mathcal{C}}
\newmathnoun{\EE}{\mathbb{E}}
\newmathnoun{\GG}{\mathbb{G}}
\newmathnoun{\sP}{\mathcal{P}}
\newmathnoun{\KK}{\mathbb{K}}
\newcommand{\gen}[1]{\mathnoun{I\left({#1}\right)}}
\newcommand{\genN}[1]{\mathnoun{I_n\left({#1}\right)}}
\newmathnoun{\BV}{\mathbb{Z}_{2^b}}
\newmathnoun{\FFp}{\mathbb{F}_p}
\newmathnoun{\FFq}{\mathbb{F}_q}
\newmathnoun{\LC}{\mathsf{LC}}
\newcommand{\smod}{%
  \mathchoice{\mskip1mu}{\mskip1mu}{\mskip5mu}{\mskip5mu}%
  \mathbin{\mathrm{smod}}%
  \mathchoice{\mskip1mu}{\mskip1mu}{\mskip5mu}{\mskip5mu}%
}

\ifanonymous
  
\else
  
\fi

\newcommand{\tabCaptionSpace}{\vspace{1ex}}

\newcommand{\gb}{Gr\"obner basis\xspace}
\newcommand{\gbs}{Gr\"obner bases\xspace}

\newcommand{\logicsize}{\fontsize{8}{9}\selectfont}
\newcommand{\qfnia}{{\logicsize\ensuremath{\mathsf{QF\_NIA}}}\xspace}
\newcommand{\qfmia}{{\logicsize\ensuremath{\mathsf{QF\_MIA}}}\xspace}
\newcommand{\qfbv}{{\logicsize\ensuremath{\mathsf{ QF\_BV}}}\xspace}
\newcommand{\qfff}{{\logicsize\ensuremath{\mathsf{QF\_FF}}}\xspace}

\newcommand{\FFT}{$\mathcal{T}_\mathtt{FF}$\xspace}
\newcommand{\TheoryBV}{$\mathcal{T}_\mathtt{BV}$\xspace}

\newcommand{\NIA}{\qfnia}
\newcommand{\OurLogic}{\qfmia}

\newcommand{\INT}{$\mathcal{T}_\mathtt{Int}$\xspace}

\newcommand{\Exp}{\ensuremath{\mathsf{Exp}}\xspace}
\newcommand{\Int}{\ensuremath{\mathsf{Int}}\xspace}
\newcommand{\Bool}{\ensuremath{\mathsf{Bool}}\xspace}

\newcommand{\Vars}{\ensuremath{X}\xspace}
\newcommand{\LVars}{\ensuremath{\mathcal{X}}\xspace}
\newcommand{\LVar}{\ensuremath{\mathsf{x}}\xspace}

\newcommand{\interp}{\ensuremath{\mathcal{I}}\xspace}

\newcommand{\toPoly}{\texttt{toPoly}}

\newmathnoun{\REqsnoc}{R^{\approx}}
\newmathnoun{\RNEqsnoc}{R^{\not \approx}}
\newmathnoun{\REqs}{R^{=}_{c}}
\newmathnoun{\RNEqs}{R^{\neq}_{c}}
\newmathnoun{\Bounds}{\texttt{B}}
\newmathnoun{\IntEqs}{R^{\approx}_{\infty}}
\newmathnoun{\IntNEqs}{R^{\not \approx}_{\infty}}
\newmathnoun{\RParam}{R^{\bowtie}_{n}}
\newmathnoun{\RParamInt}{R^{\bowtie}_{c}}
\newmathnoun{\RParamOnlyInt}{R^{\bowtie}_{\infty}}
\newmathnoun{\RMod}{R_{\text{n}}}
\newmathnoun{\RInt}{R_{\infty}}
\newmathnoun{\REqsI}{R^{\approx}_{c_i}}
\newmathnoun{\REqsN}{R^{\approx}_{n}}
\newmathnoun{\REqsNO}{R^{\approx}_{n_1}}
\newmathnoun{\REqsNi}{R^{\approx}_{n_i}}
\newmathnoun{\REqsNj}{R^{\approx}_{n_j}}
\newmathnoun{\RNEqsN}{R^{\not \approx}_{n}}
\newmathnoun{\RNEqsI}{R^{\not \approx}_{c_i}}
\newcommand{\simplify}{\ensuremath{\mathit{simp}}\xspace}

\newmathnoun{\sat}{\ensuremath{\mathsf{sat}}\xspace}
\newmathnoun{\unsat}{\ensuremath{\mathsf{unsat}}\xspace}
\newmathnoun{\unknown}{\ensuremath{\mathsf{unknown}}\xspace}

\newmathnoun{\CalcBds}{\mathtt{CalcBds}}
\newcommand{\SMTLib}{SMT-LIB\xspace}

\newmathnoun{\Reduce}{\mathtt{Reduce}}
\newmathnoun{\branch}{
  \mathchoice{\mskip1mu}{\mskip1mu}{\mskip5mu}{\mskip5mu}%
  \mathbin{\mathrm{\lvert\,\lvert}}%
  \mathchoice{\mskip1mu}{\mskip1mu}{\mskip5mu}{\mskip5mu}%
  }

\newmathnoun{\GB}{\mathtt{
GB}_{n, \le}}

\newmathnoun{\WeightedGB}{\mathtt{
GB}_{\mathit{wtd}(n,B)}}
\newmathnoun{\lm}{\mathit{lm}}
\newmathnoun{\lcm}{\mathit{lcm}}
\newmathnoun{\leadt}{\mathit{lt}}
\newmathnoun{\relu}{\mathtt{ReLU}}
\newmathnoun{\coef}{\mathtt{coef}}
\newmathnoun{\lowerbnd}{\mathtt{lb}}
\newmathnoun{\upperbnd}{\mathtt{ub}}
\newmathnoun{\tail}{\mathit{tail}}
\newmathnoun{\pnew}{e'}
\newmathnoun{\spoly}{\mathit{spoly}}

\newmathnoun{\calcBds}{\texttt{CalcBds}}
\newmathnoun{\impC}{\vDash_c}
\newmathnoun{\impN}{\vDash_n}
\newmathnoun{\nimpN}{\nvDash_n}
\newmathnoun{\nimpC}{\nvDash_c}
\newmathnoun{\nimpInt}{\nvDash_{\Int}}
\newmathnoun{\LiftPoly}{\mathtt{LP}_n^B}
\newmathnoun{\inv}{\texttt{inv}}
\newmathnoun{\liftN}{\textit{lift}_{n,B}}
\renewmathnoun{\inf}{\textit{inf}}
\newmathnoun{\lowerN}{\textit{lower}_{n,B}}

\newcommand{\dbk}[1]{\ensuremath{\llbracket{#1}\rrbracket}\xspacemm}
\newcommand{\Dbk}[1]{\ensuremath{\left\llbracket{#1}\right\rrbracket}\xspacemm}

\newcommand{\benchFfBvSp}{\texttt{f/b(s)}\xspacemm}
\newcommand{\benchFfBvMp}{\texttt{f/b(m)}\xspacemm}
\newcommand{\benchFfFfSp}{\texttt{f/f(s)}\xspacemm}
\newcommand{\benchFfFfMp}{\texttt{f/f(m)}\xspacemm}
\newcommand{\benchBvFfMp}{\texttt{b/f(m)}\xspacemm}

\newcommand{\ModExp}{\texttt{ModExp}\xspacemm}
\newcommand{\IntExp}{\texttt{IntExp}\xspacemm}
\newcommand{\GEQ}{\texttt{GEQ}\xspacemm}
\newcommand{\nGEQ}{\texttt{nGEQ}\xspacemm}
\newcommand{\LEQ}{\texttt{LEQ}\xspacemm}
\newcommand{\nLEQ}{\texttt{nLEQ}\xspacemm}
\newcommand{\OneGB}{\texttt{UnsatOne}\xspacemm}
\newcommand{\UnsatOne}{\texttt{UnsatOne}\xspacemm}
\newcommand{\UnsatDiseq}{\texttt{UnsatDiseq}\xspacemm}
\newcommand{\InconBds}{\texttt{UnsatBds}\xspacemm}
\newcommand{\UnsatBds}{\texttt{UnsatBds}\xspacemm}

\newcommand{\LiftEq}{\texttt{LiftEq}\xspacemm}

\newmathnoun{\LiftEqH}{\texttt{LiftEqH}}

\newcommand{\LiftDiseq}{\texttt{LiftDiseq}\xspacemm}
\newcommand{\LowerEq}{\texttt{LowerEq}\xspacemm}
\newcommand{\LowerDiseq}{\texttt{LowerDiseq}\xspacemm}
\newmathnoun{\ConstrBds}{\texttt{ConstrBds}}
\newcommand{\InfEq}{\texttt{InfEq}\xspacemm}
\newcommand{\RngLift}{\texttt{RngLift}\xspacemm}
\newcommand{\PrimeSqr}{\texttt{ZeroOrOne}\xspacemm}

\newmathnoun{\HeuristicGB}{H_{\mathtt{GB}}}

\newmathnoun{\HeuristicILP}{H_{\mathtt{ILP}}}

\newmathnoun{\HeuristicILPAprox}{H_{\mathtt{ILPaprox}}}

\newmathnoun{\PhiAprox}{\phi_{\mathit{aprox}}}

\newmathnoun{\PhiBoth}{\phi_{(\mathit{aprox})?}}

\newmathnoun{\slog}{\text{slog}}

\newmathnoun{\OrderBd}{<_{\mathit{Bd}}}
\newmathnoun{\Bd}{\mathit{Bd}}
\newmathnoun{\ND}{\mathit{ND}}
\newmathnoun{\OrderND}{<_{\mathit{ND}_n}}
\newmathnoun{\OrderI}{<_{I_n}}

\newnoun{\mm}{math mode}

\newmathnoun{\Split}{\mathtt{Split}}
\newmathnoun{\LiftLower}{\mathtt{Exchange}}
\newmathnoun{\CheckUnsat}{\mathtt{CheckUnsat}}
\newmathnoun{\Branch}{\mathtt{Branch}}
\newmathnoun{\lmp}{\dbk{m_{p_i}}}
\newmathnoun{\lmg}{\dbk{m_{g_i}}}
\newmathnoun{\lme}{\dbk{m_{g_i}}}

\def\orcidID#1{\smash{\href{http://orcid.org/#1}{\protect\raisebox{-1.25pt}{\protect\includegraphics{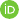}}}}}

\begin{document}

\title{
  Integer Reasoning Modulo Different Constants \\ in SMT
}
\titlerunning{
  Integer Reasoning Modulo Different Constants in SMT
}

\ifanonymous
  \author{}
  \institute{}
\else
  \author{%
    Elizaveta Pertseva\inst{1}\orcidID{0000-0001-9950-672X} \and
    Alex Ozdemir\inst{1}\orcidID{0000-0002-0181-6752} \and
    Shankara Pailoor\inst{2}\orcidID{0000-0002-9253-9585} \and
    Alp Bassa\inst{2}\orcidID{0000-0002-9685-7361} \and
    \\
    Sorawee Porncharoenwase\inst{4\thanks{work done while at Veridise}}\orcidID{0000-0003-3900-5602} \and
    Işil Dillig\inst{2,3}\orcidID{0000-0001-8006-1230} \and
    Clark Barrett\inst{1}\orcidID{0000-0002-9522-3084}
  }
  \institute{%
     Stanford University (\email{pertseva@stanford.edu})
     \and
     Veridise
     \and
        The University of Texas at Austin 
     \and 
     Amazon Web Services
  }
  \authorrunning{%
    E. Pertseva et al.
  }
\fi

\maketitle

\begin{abstract}
%

This paper presents a new refutation procedure for multimodular systems of integer constraints that commonly arise when verifying cryptographic protocols. These systems, involving polynomial equalities and disequalities modulo different constants, are challenging for existing solvers due to their inability to exploit multimodular structure. To address this issue, our method partitions constraints by modulus and uses lifting and lowering techniques to share information across subsystems, supported by algebraic tools like weighted Gröbner bases. 
Our experiments show that the proposed method outperforms existing state-of-the-art solvers in verifying cryptographic implementations related to Montgomery arithmetic and zero-knowledge proofs.

\end{abstract}

\margincommentprimer

\section{Introduction}%
\label{sec:intro}

Throughout history,
cryptosystems have been defined using modular arithmetic.
The first symmetric cipher---%
the Caesar cipher---%
used arithmetic modulo the alphabet size.
The first public key exchange
(Diffie-Hellman~\cite{diffie1976new})
used arithmetic modulo a large prime.
The first digital signature
(RSA~\cite{rivest1978method})
used arithmetic modulo a large biprime.
More recently,
cryptosystems that \textit{compute}
on secret data, such as 
homomorphic encryption~\cite{gentry2009fully},
multiparty computation~\cite{yao1982protocols},
and zero-knowledge proofs (ZKPs)~\cite{goldwasser1989knowledge}, 
often perform computation modulo a prime.

The relationship between integer arithmetic systems
with \textit{different} moduli plays
a central role in \textit{implementing} cryptography.
Generally, this is because the cryptosystem is defined using
a modulus that is different from the one natively supported by the computational model.
For example, microprocessors efficiently
perform arithmetic modulo powers of two
(e.g., $2^{16}, 2^{32}, 2^{64}, \dots$),
but elliptic-curve cryptosystems require arithmetic modulo
$\approx$256-bit primes.
To bridge the gap,
implementations use
techniques such as
Montgomery and Barrett reduction~\cite{montgomery1985modular,barrett1986implementing}.
Similar problems
(and solutions) arise in other contexts,
such as when using ZKPs.
ZKPs often only support arithmetic
modulo a large prime~\cite{walfish15cacm} and must then find ways to express and prove any properties
that are not natively defined modulo that prime.

%
%
Thus, when the correctness of a cryptosystem is expressed as a logical formula, such
as a Satisfiability Modulo Theories (SMT)
formula,  it often contains equations with different moduli.
%
%
Some
of the equations
express the cryptographic
\textit{specification}
(e.g., modulo the prime $2^{255}-19$),
and others
express the \textit{operations} performed
by the implementation
(e.g., modulo $2^{64}$).  We call such sets of constraints \emph{multimodular systems}.
%

Unfortunately, such
multimodular systems are hard for existing SMT solvers.
One approach is to use a theory that explicitly supports the modulus operator,
such as integers or bit-vectors.
In practice, this leads to poor performance because the solvers
for these theories are designed for \textit{general} modular reduction (i.e., where the modulus could be a variable)
and are not optimized
for the special case of many constraints with constant moduli.
Another approach is to combine theories
optimized for different constant moduli,
such as bit-vectors (modulo a power of two)
and finite fields (modulo a prime).
This also performs poorly because existing combination mechanisms
cannot exchange sufficiently powerful lemmas between the theories.

\begin{figure}[t]
\vspace{-0.2in}
  \centering
  \newcommand{\yoff}{-1.6}
  \newcommand{\modulesize}{0.75cm}
  \begin{tikzpicture}[scale=0.8]
  \path
    (0, 0) node[circle,draw,minimum size=\modulesize,align=center,inner sep=0] (Z) {\large $\ZZ$}
    (-2.5,\yoff) node[circle,draw,minimum size=\modulesize,align=center,inner sep=0] (Z2) {\large $\ZZ_2$}
    (   0,\yoff) node[circle,draw,minimum size=\modulesize,align=center,inner sep=0] (Z7) {\large $\ZZ_7$}
    ( 2.5,\yoff) node {\Large$\cdots$}

    (Z) edge[dotted,->, bend left=0.3cm] (Z2)
    (Z) edge[dotted,->, bend left=0.3cm] (Z7)
    (Z2) edge[dotted, ->, bend left=0.3cm] (Z)
    (Z7) edge[dotted, ->, bend left=0.3cm] (Z)

    (Z) edge[->,loop right] (Z)
    (Z2) edge[->,loop right] (Z2)
    (Z7) edge[->,loop right] (Z7)

    (4,-0.95) coordinate (guide)
    (guide) +(0.5,0.6) node[right] (alg) {algebraic reasoning}
    (guide) +(0.5,0.0) node[right] (lift) {lemma exchange}
    (lift.west) +(-0.75,0) edge[->,dotted] (lift)
    (alg.west) +(-0.75,0) edge[->] (alg)
    (lift.west) +(-0.95,0.4) coordinate (guidel)
  ;
  \draw
    (guidel |- alg.north) +(0,0.1) rectangle (alg.east |- lift.south)
  ;
  \end{tikzpicture}
  \caption{%
    Overview of our refutation procedure.
    Integer-reasoning  interacts with modules for reasoning about
    equations modulo constants, e.g.
 (in this figure, modulo 2 and 7).
    %
  }
  \label{fig:architecture}
  \vspace{-0.2in}
\end{figure}
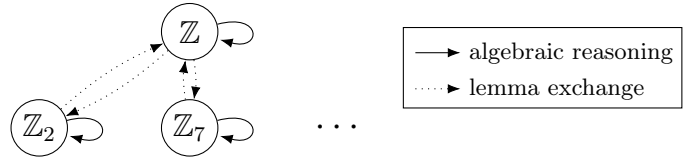

The primary contribution of this paper is a novel refutation procedure for multimodular systems.
The procedure, illustrated in Figure~\ref{fig:architecture}, partitions
constraints into \emph{subsystems} based on their moduli and then does local reasoning within each subsystem to detect conflicts.
\emph{Global reasoning} is done by exchanging lemmas across subsystems.  This is facilitated by two key
techniques: \emph{lifting}, which promotes modular constraints to the integer
domain, and \emph{lowering}, which projects integer constraints into specific
moduli.
%
%
The derivation and exchange of lemmas is further supported by two complementary strategies.
First, range analysis is used to verify whether a literal can be soundly lifted
or lowered without altering the satisfiability of the overall system.
Second, the procedure identifies additional literals \emph{implied} by a given
subsystem using algebraic techniques, such as weighted Gröbner bases and integer
linear programming, to uncover constraints that enhance cross-subsystem
reasoning.

We apply our procedure
to unsatisfiable benchmarks based on three kinds of cryptographic implementations.
The implementations include
Montgomery arithmetic~\cite{montgomery1985modular},
non-native field arithmetic for a
ZKP~\cite{succinct2024goldilocks,o1labs2023nonnative}, 
and
multi-precision bit-vector arithmetic for a ZKP~\cite{kosba2018xjsnark}.
We implement our procedure in cvc5~\cite{cvc5},
and show that it significantly outperforms prior solvers
on these benchmarks.
%
To summarize, our contributions are:
\begin{enumerate}[nosep]
  \item a refutation procedure for multimodular systems;
  \item two algorithms for finding shareable lemmas, one
    using
    a weighted monomial order
    and the other using
    integer linear constraints; and
  \item multimodular benchmarks
    from cryptographic verification applications.
\end{enumerate}

The rest of the paper is organized as follows. We
give a motivating example in Section~\ref{sec:motivating_example},
provide background in Section~\ref{sec:bg}, introduce a logic for multimodular systems in Section~\ref{sec:fragment},
explain our refutation procedure in Section~\ref{sec:dp} and its
implementation in Section~\ref{sec:implementation},
present benchmarks and experiments in Section~\ref{sec:experiments},
discuss related work in Section~\ref{sec:relwork},
and conclude in Section~\ref{sec:discuss}.

\section{Motivating Example}%
\label{sec:motivating_example}

Consider the following system of constraints $\Phi$,
which is based on code from a zero-knowledge proof
library written by Succinct Labs~\cite{succinct2024goldilocks}:
\begin{align*}
  &\Phi \triangleq&
  xy &\equiv r_1 + c_1p \mod q
     & \land &
     & 0 \le r_1, c_1 < p
     &&\land
  \\
     &&
  r_1y &\equiv r_2 + c_2p \mod q
     & \land &
     & 0 \le r_2, c_2 < p
     &&\land
   \\
     &&
  x+r_2 &\equiv r_3 + c_3p \mod q
     & \land &
     & 0 \le r_3, c_3 < p
  &&
\end{align*}
Here, $p$ and $q$ are concrete primes of $32$ and $256$ bits respectively,
but their specific values are not important.
%
%
$\Phi$ is designed to ensure a correctness condition
$C \triangleq (x+xy^2\equiv r_3 \mod p)$,
assuming bounds on the inputs $x$ and $y$: $B \triangleq (0 \le x, y < p)$.
Formally, proving $\Phi$ is correct is equivalent to proving the
validity of $(\Phi \land B) \Rightarrow C$,
or to proving
the unsatisfiability of:
\begin{equation}
  \Phi \land B \land \lnot C
  \label{eq:motivating_example}
\end{equation}

It is easy enough to refute Formula (\ref{eq:motivating_example}) by hand.
First,
observe that the range constraints in $B$ and $\Phi$
ensure that none of the equivalences
in $\Phi$ can overflow mod $q$---%
so the equivalences hold over the integers.
Second,
observe that said equivalences must also hold mod $p$:
\[
  xy \equiv r_1 \mod p
  \quad\land\quad
  r_1y \equiv r_2 \mod p
  \quad\land\quad
  x+r_2 \equiv r_3 \mod p
\]
Third,  these equivalences 
imply $x+xy^2\equiv r_3 \pmod p$,
which contradicts~$\lnot C$.

However, existing SMT solvers fail to do this refutation.
For example, when formula~(\ref{eq:motivating_example})
is encoded in \qfnia using explicit modular reductions, state-of-the-art SMT solvers (including  cvc5, z3, MathSAT, and Yices) fail to solve it. And when 
the formula is encoded in \qfbv using 512-bit bit-vectors, none of cvc5, z3, nor bitwuzla can solve it.\setcounter{footnote}{0}\footnote{All tests were run with a time limit of 20 minutes and  a memory limit of 8GB.}
\qfff solvers (cvc5 and Yices)
do not apply because they cannot encode a variable modulo more than one prime. 

The key ingredients in the manual refutation were
lifting equalities from a modular space into the integers---%
and then lowering them back to a different modular space.
In this paper, we show how to design a procedure
that performs this lifting and lowering automatically.

\section{Background}%
\label{sec:bg}

In this section,
we define notation and provide a brief overview of
algebra~\cite{mccoy1948rings},
ideals~\cite{david1991ideals},
and
SMT~\cite{BT18}.
%
More details can be found in the cited work.
\subsection{Algebra}

\paragraph{Sets, Intervals, and Functions}
%
%
Let $\ZZ$ be the integers, \ZZNN the non-negative integers, and $\ZZ^{+}$ the positive integers. Let $\ZZ^{+}_{\infty}$ be $\ZZ^{+}\cup\{\infty\}$,
and let $\ZZ_n$ be the non-negative integers less than $n$.
%
$\Vars$ denotes the set of variables $\{x_1, \dots, x_k\}$.
%
For a set $S$ and some $t$, the notation $S$,$t$ abbreviates $S \cup \{t\}$.
$[i,j]$ denotes the closed integer interval from $i$ to $j$.
Interval intersection is defined in the usual way: $[a,b] \cap [c,d] = [\max(a,c), \min(b,d)]$.  
%
For an interval, pair, or sequence $t$, we denote by $t_i$  the $i^{\mathit{th}}$ element of $t$, e.g., $(a,b)_1 = a$.
%
%
We use two variants of the modulo function.
$a \bmod n$, for $a \in \ZZ$ and $n \in \ZZ^{+}$,
is the unique $r \in \ZZ_n$ such that $a = qn + r$ for some $q\in\ZZ$
(as in SMT-LIB~\cite{BarFT-SMTLIB}).
The signed variant $a \smod n$ is defined as $a \bmod n$ if that value is at
most $\frac{n}{2}$ and $(a \bmod n) - n$ otherwise.

\paragraph{Polynomial Rings} 
%
Let $R$ be a ring~\cite{dummit_and_foot}.
We overload $\ZZ$ to also denote the integer ring
and $\ZZ_n$ to also denote the ring over $\{0, \dots, n-1\}$, with addition and multiplication modulo $n$.
Both $\ZZ$ and $\ZZ_n$ are principal ideal rings (PIR), and if $n$ is prime, $\ZZ_n$ is also a \emph{field}.
Let $R[\Vars]$ denote the ring of polynomials with variables in $\Vars$ and coefficients in $R$.
%
%
%
In this paper, we focus on the following rings: (1) $\ZZ[\Vars]$, which is the ring of polynomials with integer coefficients, and (2) $\ZZ_n[\Vars]$, which is the ring of polynomials with integer coefficients modulo $n$.
%
%
%
In ring $R[\Vars]$, a \emph{monomial} is a polynomial of the form
$x^{e_1}_1\cdots x_k^{e_{k}}$, with $e_i\in \ZZNN$.
When $e_i=0$ for every $i$, we denote the monomial as $1$.
A \emph{term} is a monomial multiplied by a coefficient.
Polynomials are written as a sum of terms with distinct monomials.

\paragraph{Monomial Orders}
A \emph{monomial order} $\le$ is a total order on monomials
that satisfies the following properties:
($i$) $1 \le m$ for every monomial $m$;
and ($ii$) for all monomials $m_1, m_2, m$, if $m_1\leq m_2$, then $m_1m \leq m_2m$.
Examples of monomial orders for a monomial of the form $x^{e_1}_1\cdots x_k^{e_{k}}$ include the following:
\emph{lexicographic order} compares monomials by (lexicographically) comparing their exponent tuples $(e_1,\dots,e_k)$,
\emph{graded reverse lexicographic order} by comparing $(\sum_{i=1}^ke_i, -e_k, \ldots, -e_1 )$,
and \emph{weighted reverse lexicographic order} by comparing
$(\sum_{i=1}^kw_ie_i, -e_k, \dots, -e_1)$
for a fixed tuple of weights $(w_1 \hdots w_k)$.
Given a monomial order,
the \emph{leading monomial} of a polynomial $p$, denoted $\lm(p)$,
is the largest monomial occurring in $p$ with respect to the monomial order.
The leading term, denoted $\leadt(p)$, is that monomial's term.
%
%

\subsection{Ideals}
For a set of polynomials $S = \{f_1, \dots, f_n\} \subset R[\Vars]$,
$I(S) = \{g_1f_1 + \dots + g_nf_n \mid g_i \in R[\Vars]\}$ is the \emph{ideal generated by $S$}.
%
%
In order to disambiguate which ring $R$ is meant in the definition of an ideal,
we use the notation $I_n(S)$, with $n\in\ZZ^{+}_{\infty}$.  The meaning of $I_n(S)$ is either the ideal generated by $S$ with $g_i \in \ZZ_n[\Vars]$, when $n\in\ZZ^{+}$, or the ideal generated by $S$ with $g_i \in \ZZ[\Vars]$, when $n = \infty$.

A \textit{solution} to the polynomial system $S \subset R[\Vars]$,
is  $\vec a \in R^{k}$ such that for all $f \in S$, $f(a_1, \dots, a_k)=0$.
The set of all solutions is called the \emph{variety} of $S$,
denoted  $\mathcal{V}(S)$.
As above, we use the subscript $n$ to distinguish among rings.
Thus $\mathcal{V}_n(S)$ for $n\in\ZZ^{+}$
is a subset of $\ZZ_n^k$
and
$\mathcal{V}_{\infty}(S)$
is a subset of $\ZZ^k$.
If $1\in I_n(S)$, then $\mathcal{V}_n(S) = \emptyset$, i.e., $S$ has no solution.
However, the converse does not hold.

One incomplete test for ideal membership is \emph{reduction}.
%
For the polynomials $p$, $g$, and $r$ $\in R[\Vars]$, where $R$ is a PIR, $p$ \emph{reduces} to $r$ modulo $g$, written $p \rightarrow_g r$, if some term $t$ of $p$ is divisible by $\leadt(g)$ with $r = p -\frac{t}{\leadt(g)} g$ \cite{eder2021efficient}. 
If $R$ is a field, then $p \rightarrow_g r$, if some term $t$ of $p$ is divisible by $\lm(g)$.
Reduction is also defined for a set of polynomials $S$, written as $p \rightarrow_S r$. $p$ reduces to $r$ modulo $S$ if there is a sequence of reductions from $p$ to $r$, each modulo some polynomial in $S$, and no further reduction of $r$ modulo $S$ is possible.
If $p$ reduces to 0 modulo $S$ then $p$ belongs to the ideal generated by $S$.
However, once again, the converse does not hold.

\paragraph{\gbs} A \emph{\gb}~\cite{buchberger} is a set of polynomials with special properties, including that reduction is a complete test for ideal membership: 
$p$ reduces to 0 modulo a \gb iff $p$ belongs to the ideal generated by the \gb.
There exist numerous algorithms for computing \gbs, including Buchberger's algorithm~\cite{buchberger}, F4~\cite{f4}, and F5~\cite{f5}.
%
%
%
%
%
We use $\GB$ with $n\in \ZZ_{\infty}^{+}$ to refer to a \gb computation. 
The subscript $n$ indicates which ring the \gb is computed in ($n\in \ZZ^{+}$ means $\ZZ
_n[\Vars]$ and $n = \infty$ means $\ZZ[\Vars]$), while
$\le$ indicates which monomial order to use. 
In this paper, we assume \gbs are \emph{strong} and \emph{reduced}, meaning that \GB is always deterministic, producing a single unique basis.

\subsection{SMT}
In addition to the algebraic domains above, we also work in the logical setting of many-sorted first-order logic with
equality~\cite{enderton2001mathematical}.
$\Sigma$ denotes a signature with a set of sort symbols (including \Bool), a symbol family $\approx_\sigma$  with sort $\sigma \times \sigma \rightarrow \Bool$ for all sorts $\sigma \in \Sigma$,\footnote{We drop the $\sigma$ subscript when it is clear from context.} and a set of interpreted function symbols.
We assume the usual definitions of well-sorted $\Sigma$-terms and literals, and refer to $\Sigma$-terms of sort \Bool as formulas.
To distinguish logical $\Sigma$-terms from algebraic terms in polynomials (defined above), we write \emph{$\Sigma$-term} to refer to the former, where $\Sigma$ is the signature.
A \emph{theory} is a pair $\mathcal{T} = (\Sigma, \textbf{I})$, where $\Sigma$ is a
signature and $\textbf{I}$ is a class of $\Sigma$-interpretations.
A \emph{logic} is a theory together with a syntactic restriction on formulas.
A formula $\phi$ is \emph{satisfiable} if it evaluates to \texttt{true}
in some interpretation in \textbf{I}.
Otherwise,
$\phi$ is \emph{unsatisfiable}.

The CDCL($\mathcal{T}$) framework of SMT aims to determine if a formula $\phi$ is \emph{satisfiable}. 
At a high level, a \emph{core} module explores the propositional abstraction of $\phi$ and forwards literals corresponding to the current propositional assignment to the \emph{theory solver}.
A \emph{theory solver}, specialized for a particular theory $\mathcal{T}$, checks if there exists an interpretation that satisfies the received set of literals. 
We focus on three main theories.
The theory of finite fields (defined in~\cite{CAV:OKTB23}), which we refer to as \FFT, reasons about finite fields, i.e., rings $\ZZ_n[\Vars]$ where $n$ is prime. 
We also make use of the standard SMT-LIB~\cite{BarFT-SMTLIB} theories of bit-vectors, which we denote \TheoryBV, and integer arithmetic, which we denote \INT. 
For two \INT terms $s$, $t$, we abbreviate the literal $\neg (s \approx t)$ as $s\not\approx t$. 
We use $\bowtie$ to refer to an operator in the set $\{\approx, \not\approx\}$.
\NIA refers to the \SMTLib logic that uses the theory \INT and restricts formulas to be quantifier-free.
%


\section{A Multimodular Logic}%
\label{sec:fragment}

\lstset{
  basicstyle=\itshape,
  xleftmargin=3em,
  literate={->}{$\rightarrow$}{2}
           {α}{$\alpha$}{1}
           {δ}{$\delta$}{1}
}

\begin{figure}[t]    
\begin{subfigure}[b]{0.4\textwidth}
  \begin{tabular}{@{}l@{\hskip 10pt}l@{\hskip 10pt}}
    \toprule
     Symbol & Arity \\
    \midrule
    & \\
    $n \in \ZZ $ &  \Int \\
    $-,+,\times$ &  \Int   $\times$ \Int $\rightarrow$ \Int  \\
    $\bmod$ & 
    \Int   $\times$ \Int $\rightarrow$ \Int   \\
    $\approx$ & \Int   $\times$ \Int $\rightarrow$ \Bool   \\
    $\leq,\geq$ & \Int $\times$ \Int $\rightarrow$ \Bool  \\
    & \\
    \bottomrule
  \end{tabular} 
  \caption{Signature used by \OurLogic.}
  \label{fig:signature}
\end{subfigure}
\hspace{0.05\textwidth}%
\begin{subfigure}[b]{0.54\textwidth}
\[
\begin{aligned}
  \texttt{Op} &\rightarrow \times \; | \; \texttt{+} \; | \; \texttt{-} \\
  \texttt{Exp} &\rightarrow ( \texttt{Exp} \; \texttt{Op} \; \texttt{Exp} )\; |  \; \texttt{Var} \;  | \; \texttt{Int} \\
  \texttt{BExp} &\rightarrow \texttt{Var} \; \texttt{$\leq$} \; \texttt{Int} \;
  | \; \texttt{Var} \; \texttt{$\geq$} \; \texttt{Int}\\
\texttt{EqExp}  &\rightarrow \texttt{Exp} \; \approx \; 0 \; | \;
    \texttt{Exp} \;  \bmod \; \texttt{Int} \; \approx \; 0\\
  \texttt{Atom}
   &\rightarrow \texttt{BExp} \; | \; \texttt{EqExp} \\
  \texttt{Literal} &\rightarrow \texttt{Atom} \; | \; \neg \texttt{Atom}
\end{aligned}
\]
\caption{\OurLogic grammar; \texttt{Int} $\in\ZZ$ and \texttt{Var} $\in \Vars$.}
\label{fig:fragment}
\end{subfigure}

\caption{The signature and grammar for \OurLogic, a fragment of \NIA.}
\vspace{-0.2in}
\end{figure}

%
%
%
Previous work on verifying arithmetic modulo large primes
\cite{CAV:OKTB23, CAV:OPBFBD24}
encodes constraints using \FFT.
However, the signature of \FFT does not support non-prime moduli or constraints
that share variables and use different moduli, limiting the range of
problems that can be encoded.
Instead, we encode multimodular constraints directly in \INT. 
%
%
However, we restrict the syntax of \INT by defining a logic called \OurLogic (multimodular integer arithmetic), which is
a fragment of \qfnia. 
%
Importantly, \OurLogic is not \textit{semantically} weaker than \qfnia;
it includes integer polynomials and predicates,
so all of \qfnia can be encoded in \OurLogic via standard rewrites.
Rather,
\OurLogic is a syntactic restriction of \NIA
that is designed to make the multimodular structure of queries
clearer
so that our procedure can leverage that structure.

The subset of the signature of \INT used by \OurLogic is shown in Figure~\ref{fig:signature}.  From now on, we use $\Sigma$ to denote this signature.  A grammar for the syntactic fragment of \qfmia is shown in Figure~\ref{fig:fragment}. We assume a set $\LVars = \{\LVar_1,\dots,\LVar_k\}$ of logical variables of sort \Int and define corresponding categories of $\Sigma$-terms as follows.
A \OurLogic \emph{expression} (produced by \texttt{Exp}) is either an integer constant, a variable (in $\LVars$), or an application of one of the operators $-$, $+$, or $\times$ to two expressions.  For simplicity, we often represent multiplication with juxtaposition (e.g., $ab$ instead of $a\times b$) and leave out parentheses when clear from context (i.e., when they can be inferred from standard operator precedence rules).
A \OurLogic \emph{atom} (produced by \texttt{Atom}) is an inequality (between a variable and a constant), an equality between an expression and 0, or an equality between an expression modulo some integer constant and zero.
A \OurLogic \emph{literal} (produced by \texttt{Literal}) is either an atom or the negation of an atom.  From now on, unless otherwise noted, expressions, atoms, and literals are \OurLogic expressions, atoms, and literals.  We also assume  that interpretations are \INT-interpretations and all notions of satisfiability are modulo \INT. 
%
Since we need to work in both the algebraic and the logical domains, we assume a bijection from logical variables in $\LVars$ to algebraic variables in $\Vars$ and define an operator \dbk{\cdot} that takes a logical $\Sigma$-term and returns the corresponding polynomial in $\ZZ[\Vars]$ (distributing multiplication and combining like terms as necessary to obtain a sum of terms, each with a unique monomial).

\section{Refutation Procedure}%
\label{sec:dp}

In this section, we describe our refutation procedure.
First we describe our approach at a high level, and then
we discuss its key technical ingredients.

\subsection{Key Ideas}


\label{subsec:overview}

A naive approach to solving a multimodular system of constraints is to use a standard encoding in \qfnia, introducing an auxiliary variable for each modular constraint.
For instance, $x \equiv y \ (\textsf{mod} \  n)$ would be encoded as $x \approx y+n\cdot k$ using an auxiliary integer variable $k$.
While sound,
this naive approach scales poorly,
as we show experimentally in Section~\ref{sec:experiments}.
Our key insight is that this limitation can often be overcome by \emph{partitioning} the original system into a \emph{set}
of different subsystems, one for each specific modulus.
%
Reasoning in each subsystem can be done efficiently, and if any subsystem is
unsatisfiable, then so is the original set of constraints.
However, since the converse is not true, our procedure seeks to exchange as much
information as possible between the different subsystems, with the goal of  improving our ability to detect unsatisfiable constraints.
%
To enable the exchange of information between the different subsystems, we employ the concepts of
\emph{lifting} and \emph{lowering}.

\begin{definition}[Liftable]
  Let $C$ be a set of \OurLogic literals.
  A literal of the form $e \bmod n \bowtie 0 $ is liftable (in $C$)
  if $C \cup  \{ e \bmod n \bowtie 0 \}$
  is equisatisfiable to
  $C \cup \{ e \mod n \bowtie 0,  e \bowtie 0\}$.
\end{definition}
\noindent
In other words, a constraint is liftable if adding the constraint without the
modulus $n$ maintains equisatisfiability.

\begin{definition}[Lowerable] Let $C$ be a set of \OurLogic literals. A literal the form $e \bowtie 0$
is \emph{lowerable} (in $C$) with respect to $n$ if $C \cup \{ e\bowtie 0 \}$ and $C \cup \{ e \bowtie 0, e \bmod n \bowtie 0 \}$ are equisatisfiable.
\end{definition}
\noindent
In other words, a constraint is \emph{lowerable} w.r.t. $n$
if adding it with a modular reduction maintains equisatisfiability.

\begin{remark}
    \label{remark:new_definition}
    If the literal $e \bmod n \bowtie 0$ is implied by $C$, then the lifting definition reduces to: $C$ and $C \cup \{e \bowtie 0\}$ must be equisatisfiable. Similarly, if the  literal  $e \bowtie 0$ is implied by $C$, it is lowerable w.r.t. $n$ if $C$ and $C\cup \{e  \bmod n \bowtie 0\}$ are equisatisfiable. In the remainder of the paper, we rely on these simpler versions of the definitions.
\end{remark}

Lifting provides a way for each subsystem containing modular constraints to share information in a common language without adding new variables.
%
Lowering adds constraints to individual subsystems and can often result in significant simplifications: when we lower an equation with respect to $n$, all integer constants divisible by $n$ can be replaced by $0$.
For example, lowering $\LVar_1 - 6\LVar_2 \approx 0$ with respect to $6$
adds a new equality $\LVar_1 \approx  0$ to the subsystem with modulus $6$.  We denote by $\simplify_n(e)$ the result of replacing every integer constant $c$ in $e$ by $c \smod n$ \footnote{We use $\smod$ instead of $\bmod$ to reduce the magnitude of coefficients of $\simplify_n(e)$ and increase the likelihood  that $\simplify_n(e)$ is liftable according to Lemma~\ref{lem:eq_lift}.} and simplifying. 

\newcommand{\CEqsN}{\ensuremath{\mathcal{C}_n^=}\xspace}
\newcommand{\CNEqsN}{\ensuremath{\mathcal{C}_n^{\ne}}\xspace}

As expected, not all constraints are liftable or lowerable.
In order to facilitate the inference of liftable and lowerable constraints,
our method splits the constraint system into the following subsystems,
for $n \in \PosZZ$:
\begin{itemize}[leftmargin=*]
\item
  {\bf Modulus-$n$ equality subsystems}
  are sets of expressions $\{e_1, \dots, e_m\}$. Each $e_i$ represents the constraint $e_i\bmod n \approx 0$ when $n\neq \infty$
  or the constraint $e_i \approx 0$ when $n = \infty$.
  We write $R_n^{\approx}$ to denote the set of expressions.
\item
  {\bf Modulus-$n$ disequality subsystems}
  are sets of expressions $\{e_1, \dots, e_m\}$. Each $e_i$ represents the constraint $e_i \bmod n  \not \approx 0$ when $n\neq \infty$
  or the constraint $e_i \not \approx 0$ when $n = \infty$.
  We write $R_n^{\not \approx}$ to denote the set of expressions.
\item {\bf Variable bounds:}
  Our method also maintains a mapping $B$ from each variable in $\LVars$ to its lower and upper bound (with $-\infty/\infty$ denoting unbounded variables).
  As we will see shortly, this bound information is crucial for identifying lowerable and liftable equations.
\end{itemize}

\begin{figure}[t]

\begin{center}
\AxiomC{ $\bowtie \ \in \{ =, \neq \}$}
  \AxiomC{$(e \bmod n \bowtie 0) \in C$}
\LeftLabel{\ModExp}
\BinaryInfC{$C:= C \setminus \{ e \bmod n \bowtie 0\} \quad \RMod^{\bowtie} := \RMod^{\bowtie}, \simplify_n(e)$}
  \DisplayProof
\end{center}
\begin{center}
\AxiomC{ $\bowtie \ \in \{ =, \neq \}$}
 \AxiomC{$ (e\bowtie  0) \in C$}
\LeftLabel{\IntExp}
\BinaryInfC{ $C:= C\setminus \{e \bowtie 0\} \quad R_{\infty}^{\bowtie} := R_{\infty}^{\bowtie} , e$}
  \DisplayProof
\end{center}
\begin{center}
\AxiomC{$(\LVar\geq n) \in C$}
\LeftLabel{\GEQ}
\UnaryInfC{$C:= C\setminus \{\LVar \geq n\} \quad B(\LVar)_1$ := $\max(B(\LVar)_1, n)$}
  \DisplayProof
\end{center}

\begin{center}
\def\defaultHypSeparation{\hskip .01in}
\AxiomC{$(\LVar \leq n) \in C$ }
\LeftLabel{\LEQ}
\UnaryInfC{$C:= C\setminus \{\LVar \leq c\} \quad B(\LVar)_2$ := $\min(B(\LVar)_2, n) $}
  \DisplayProof
\end{center}

\begin{center}
\AxiomC{$\neg(x\geq n) \in L$}
\LeftLabel{\nGEQ}
\UnaryInfC{$C:= C\setminus \{\neg(\LVar \geq n)\} \quad B(\LVar)_2 := \min(B(\LVar)_2, n-1)$}
  \DisplayProof
\end{center}

\begin{center}
\AxiomC{$\neg(\LVar\leq n) \in L$}
\LeftLabel{\nLEQ}
\UnaryInfC{$C:= C\setminus \{\neg(\LVar \leq n)\ \quad B(\LVar)_1:= \max(B(\LVar)_1, n+1)$}
  \DisplayProof
\end{center}
\caption{Encoding rules for a multimodular system $C$, where $e$ is an  expression and $n\in\ZZ^{+}_{\infty}$.}
\label{fig:split}
\end{figure}

Given a multimodular system $C$, we assume that it is encoded as a tuple $(B, \IntEqs, \IntNEqs, R_{n_1}^{\approx}, R_{n_1}^{\not \approx}, \dots, R_{n_k}^{\approx}, R_{n_k}^{\not \approx})$, where $\{n_1,\dots,n_k\}$ is the set of all integer constants greater than $1$ appearing anywhere in a literal in $C$.
For example, for a multimodular system  $\{2x \bmod 3 \approx  0, x < 5 \}$, the resulting tuple would be $(B,  \IntEqs, \IntNEqs,  R_{2}^{\approx}, R_{2}^{\not \approx}, R_{3}^{\approx}, R_{3}^{\not \approx},R_{5}^{\approx}, R_{5}^{\not \approx} \}$.
A set of rules for accomplishing the encoding is included in Figure~\ref{fig:split}.

Going forward, we use $C$ to refer both to the original set of constraints and to the tuple encoding.  
%
We also define $\CalcBds(B,e)$ as a function that returns the maximum and minimum possible values that expression $e$ can take when evaluated at the variable assignments permitted by the variable bounds map $B$, based on standard interval arithmetic \cite{intervalarithmetic}.  For example, given an equality $\LVar_1 \LVar_2 \bmod 6 \approx 0$ and variable bounds $B = \{\LVar_1 : [0,6], \LVar_2: [0,6]$\}, $\CalcBds(B, \LVar_1\LVar_2)$ would return the  interval $[0,36]$. 
Using this machinery, we now state the following lemmas that help identify liftable and lowerable constraints. We include the proofs in Appendix~\ref{sec:lemmaproofs}. Recall that $I_n$ computes the ideal generated by a set of polynomials.

\begin{lemma}
Let $C$ be a multimodular system with bounds $B$ and modulus-$n$ equality subsystem $\REqsN$, with $n\in\ZZ^{+}$.
Then, an equality $e \bmod n \approx 0$ is liftable in $C$ if
(1) $\dbk{e}
\in I_n(\dbk{\REqsN})$
and
(2) $\CalcBds(B, e) \subseteq [1-n, n-1]$.
\label{lem:eq_lift}
\end{lemma}

\noindent
This lemma is useful in two ways.
First, if $e$ is in $\REqsN$, then certainly, $\dbk{e}$ is in $I_n(\dbk{\REqsN})$;
thus, checking whether the constraint represented by
 an element of $\REqsN$ is liftable reduces to computing $\CalcBds(B,e)$,
which can be done in linear time.
%
Second, this lemma gives a way to find \emph{additional} liftable
equalities that are not part of the original constraint system by identifying
polynomials that are in the ideal of $\dbk{\IntEqs}$.
While listing all the polynomials in an ideal is infeasible,
later subsections (Sections~\ref{subsec:gb} and~\ref{subsec:ilp}) explore effective methods to identify useful polynomials that are in
the ideal.
The next lemma states that disequalities are always liftable.

\begin{lemma}
Let $C$ be a multimodular system with modulus-$n$ disequality subsystem $\RNEqsN$, with $n\in\ZZ^{+}$.
Then a disequality $e\bmod n \not \approx 0$ is liftable in $C$ w.r.t. $n$ if $e\in \RNEqsN$.
\label{lem:diseq_lift}
\end{lemma}

\noindent
This lemma states that all of the original disequalities are liftable. However, unlike the equality case, we cannot infer \emph{additional}
disequalities using ideals, as disequalities are not preserved under the operations used to construct the elements of an ideal.
The next two lemmas are dual to Lemmas~\ref{lem:eq_lift}
and~\ref{lem:diseq_lift}, but are for lowerability instead of liftability.

\begin{lemma}
  Let $C$ be a multimodular system with integer equalities $\IntEqs$.
  If $e$ is an expression and $\dbk{e} \in  I_{\infty}(\dbk{\IntEqs})$, then $e\approx0$ is lowerable w.r.t. every $n\in\ZZ^{+}$.
  \label{lem:eq_lower}
\end{lemma}
\begin{lemma}
    Let $C$ be a multimodular system with integer disequalities $\IntNEqs$ and bounds $B$.
    Then, if $n\in\ZZ^{+}$, a disequality $e \not \approx 0$ is lowerable in $C$ w.r.t. to $n$ if (1) $e \in \IntNEqs$ and (2) $\CalcBds(B,e) \subseteq [1-n,n-1]$ 
\label{lem:diseq_lower}
\end{lemma}

\subsection{Refutation Calculus}
\label{subsec:calculus}

Next, we  leverage the notions of liftability and lowerability defined in  the previous subsection to formulate our \emph{refutation calculus} (presented in Fig.~\ref{fig:calculus}).
The rules in our calculus serve four main roles.
First, they establish whether a specific subsystem is unsatisfiable.
Second, they attempt to tighten existing bounds and learn new equalities from these bounds.
Third, they leverage Lemmas~\ref{lem:eq_lift}--\ref{lem:diseq_lower} to
exchange information between different subsystems via lifting and lowering.
Finally, they simplify unliftable equalities via branching.

\begin{figure}[H]
 \begin{center}
 \def\defaultHypSeparation{\hskip .01in}
  \AxiomC{$1 \in I_n(\dbk{\REqsN})$}
  \LeftLabel{\OneGB}
  \UnaryInfC{\unsat}
  \DisplayProof
\quad 
  \AxiomC{ $ e \in \RNEqsN $}
 \AxiomC{$\dbk{e} \in I_n(\dbk{\REqsN})$}
  \LeftLabel{\UnsatDiseq}
  \BinaryInfC{\texttt{unsat}}
  \DisplayProof
\end{center}
\begin{center}
\def\defaultHypSeparation{\hskip .01in}
 \AxiomC{ $\; B(\LVar_i)_1 > B(\LVar_i)_2$}
 \LeftLabel{\InconBds}
 \UnaryInfC{\unsat}
 \DisplayProof
\end{center}
\begin{center}
\begingroup
\footnotesize
\def\defaultHypSeparation{\hskip .01in}
  \AxiomC{ $\dbk{e} \in I_{\infty}(\dbk{\IntEqs})$}
  \AxiomC{$e = a \LVar_i + e'$}
  \AxiomC{$x_i \notin e'$}
  \AxiomC{$B(\LVar_i)_1 \leq B(\LVar_i)_w$}
 \LeftLabel{\ConstrBds}
 \QuaternaryInfC{$B(\LVar_i):= \CalcBds(B,-\frac{1}{a}(e-a\LVar_i)) \cap B(\LVar_i) $}
 \DisplayProof
\endgroup
\end{center}

\begin{center}
  \AxiomC{ $B(\LVar_i)_1=B(\LVar_i)_2$}
  \AxiomC{$ \dbk{\LVar_i - B(\LVar_i)_1}  \notin I_{\infty}(\dbk{\IntEqs})$}
 \LeftLabel{\InfEq}
 \BinaryInfC{$\IntEqs:= \IntEqs, \LVar_i - B(\LVar_i)_1$ }
 \DisplayProof
\end{center}

\begin{center}
\def\defaultHypSeparation{\hskip .01in}
 \AxiomC{$n \neq \infty$}
 \AxiomC{$\dbk{e} \in I_n(\dbk{\REqsN})$}
 \AxiomC{$\CalcBds(B,e) \subseteq [1-n,n-1]$}
 \AxiomC{$ \dbk{e} \notin I_{\infty}(\dbk{\IntEqs}) $}
 \LeftLabel{\LiftEq}
 \QuaternaryInfC{$\IntEqs := \IntEqs, e$ }
 \DisplayProof
 \end{center}

\begin{center}
\def\defaultHypSeparation{\hskip .01in}
\AxiomC{$n\neq \infty$}
\AxiomC{$e\in \RNEqsN$}
\AxiomC{$e \notin \IntNEqs$}
\LeftLabel{\LiftDiseq}
\TrinaryInfC{$\IntNEqs := \IntNEqs,e$ }
\DisplayProof
\end{center}
\begin{center}
\def\defaultHypSeparation{\hskip .01in}
\AxiomC{$\dbk{e} \in I_{\infty}(\dbk{\IntEqs})$}
\AxiomC{$n \neq \infty$}
 \AxiomC{$\dbk{\simplify_n(e)} \notin I_{n}(\dbk{\REqsN})$}
\LeftLabel{\LowerEq}
\AxiomC{$\REqsN \in C$}
\QuaternaryInfC{$\REqsN := \REqsN, \simplify_n(e)$ }
\DisplayProof
\end{center}

\begin{center}
\begingroup
\footnotesize
\def\defaultHypSeparation{\hskip .01in}
  \AxiomC{ $e \in$ \IntNEqs}
  \alwaysNoLine
  \UnaryInfC{$n \neq \infty$ {\hskip .03in}
 $\CalcBds(B,e) \subseteq [1-n,n-1]$ {\hskip .03in} $ \simplify_n(e) \notin$ \RNEqsN {\hskip .03in} $\RNEqsN \in C$}
 \LeftLabel{\LowerDiseq}
 \alwaysSingleLine
 \UnaryInfC{$\RNEqsN := \RNEqsN, \simplify_n(e)$ }
 \DisplayProof
\endgroup
\end{center}
\begin{center}
\def\defaultHypSeparation{\hskip .01in}
  \AxiomC{$ \dbk{e} \in I_{n}(\dbk{\REqsN})$}
  \AxiomC{$n \neq \infty$}
  \alwaysNoLine
  \def\extraVskip{2pt}
  \BinaryInfC{$\CalcBds(B,e) \not\subseteq [1-n,n-1]$ {\hskip .03in} $\CalcBds(B,e) \subseteq [1-2n, 2n-1]$ }
  \def\extraVskip{2pt}
  \UnaryInfC{ $ \dbk{e-n} \notin I_{\infty}(\dbk{\IntEqs})$ {\hskip .03in} $ \dbk{e} \notin I_{\infty}(\dbk{\IntEqs}) $ {\hskip .03in} $ \dbk{e+n} \notin I_{\infty}(\dbk{\IntEqs})$}
\LeftLabel{\RngLift}
\alwaysSingleLine
\UnaryInfC{$\IntEqs := \IntEqs, e - n \branch \IntEqs := \IntEqs, e \branch \IntEqs := \IntEqs, e+n$}
  \DisplayProof
\end{center}
\begin{center}
\def\defaultHypSeparation{\hskip .01in}
  \AxiomC{$ \dbk{e}  \in I_{n}(\dbk{\REqsN})$}
  \AxiomC{$ \dbk{e} = s^2 -s$}
   \alwaysNoLine
  \def\extraVskip{2pt}
  \BinaryInfC{ $n \neq \infty$ {\hskip .03in} \texttt{IsPrime}($n$)  {\hskip .03in} $ \dbk{s} \notin I_{n}(\dbk{\REqsN})$ {\hskip .03in} $\dbk{s-1} \notin I_{n}(\dbk{\REqsN}) $ }
  \LeftLabel{\PrimeSqr}
  \alwaysSingleLine
  \UnaryInfC{ $\REqsN := \REqsN, s \branch \REqsN := \REqsN, s-1$}
  \DisplayProof
\end{center}
\vspace{-0.2in}
\caption{ Derivation rules.  $e,s$ are  expressions, $a\in \ZZ$, and $n\in\ZZ^{+}_{\infty}$.}
\label{fig:calculus}
\end{figure}

We present the calculus as rules that modify \emph{configurations},
as is common in SMT procedures \cite{sheng2022reasoning, liang2014dpll}.
Here, a configuration is the representation of the system of constraints 
$C$ as the tuple $(B, \IntEqs, \IntNEqs, R_{n_1}^{\approx}, R_{n_1}^{\not \approx}, \dots, R_{n_k}^{\approx}, R_{n_k}^{\not \approx})$ as described
in Section~\ref{subsec:overview}.
The rules are presented in \emph{guarded assignment form},
where the premises describe the conditions on the current configuration under which the rule can be
applied,
and the conclusion is either \unsat
or indicates how the configuration is modified.
A rule may have multiple alternative
conclusions separated by \branch. 
An application of a rule is \emph{redundant} if it does not change the
configuration in any way.
A configuration other than \unsat is \emph{saturated} if every possible
application is redundant. 
A \emph{derivation tree} is a tree where each node is a configuration and its children, if any, are obtained by a non-redundant application of a rule of the calculus. A derivation tree is closed
if all of its leaves are \unsat.  A \emph{derivation} is a sequence of derivation trees in which each element in the sequence (after the first) is obtained by expanding a single leaf node in the previous tree. We explain each class of rules in the calculus in more detail below.
%
%


\paragraph{Checking Unsatisfiability}
\UnsatOne is used to conclude unsatisfiability using algebraic techniques: As explained in Section~\ref{sec:bg},  if some ideal \dbk{\REqsN} contains the polynomial 1, this  indicates that the constraints represented by \REqsN have no common solution in the modulus-$n$ subsystem. Hence, we can conclude that the whole system is unsatisfiable.
\UnsatDiseq checks if any of the polynomials in the ideal of \dbk{\REqsN} match an expression in \RNEqsN. 
%
%
Finally, \InconBds concludes \unsat if some variable's lower bound exceeds its upper bound.

\paragraph{Tightening Bounds}
\ConstrBds tightens bounds based on equations in \IntEqs.
Suppose $e$ is an expression, with $\dbk{e} \in I_{\infty}(\dbk{\IntEqs})$, consisting of the sum of $a\LVar_i$ and some other expression $e'$.
%
\ConstrBds uses $e$ to compute new bounds for $\LVar_i$ and then intersects these with the current bounds for $\LVar_i$.
%
%
%
 \InfEq uses the bounds information to infer new equalities: When a variable's upper and lower bounds become equal, we can obtain a new equality.

\paragraph{Lifting/Lowering }
The lifting and lowering rules are directly based on the lemmas stated in Section~\ref{subsec:overview}. 
The rules \LiftEq and \LiftDiseq \emph{lift}
(dis)equalities from $\RMod^{\bowtie}$ to $\RInt^{\bowtie}$, by relying on Lemmas~\ref{lem:eq_lift}~and~\ref{lem:diseq_lift}. 
We also check entailment (polynomial ideal containment for equalities and set containment for disequalities) to avoid adding redundant information.
Dually,  the rules  \LowerEq and \LowerDiseq \emph{lower} (dis)equalities from  from $\RInt^{\bowtie}$ to  $\RMod^{\bowtie}$ by relying on Lemmas~\ref{lem:eq_lower}~and~\ref{lem:diseq_lower}. 
\paragraph{Branching Rules} When encountering equations that are \emph{almost} liftable, we rely on branching rules to either lift or simplify them. For example, if $\CalcBds(B,x) = [0,7]$, $x \bmod 6 \approx 0$ is \emph{almost} liftable, as its range exceeds the liftable range only slightly. 
\RngLift branches on possible values for $e$ that are within $|n|$ of the liftable range.
\PrimeSqr detects when the
value of a variable can only be 0 or 1;
it applies only to prime moduli
because other moduli have zero divisors.

\label{subsec:correctness}
We state \emph{soundness} and \emph{termination} theorems for our  calculus and include the proofs in Appendix~\ref{ap:sound_term}.
However, the calculus is \emph{not} complete---a saturated leaf in a derivation tree
does not necessarily mean that the constraints at the root of the tree are satisfiable.
%
%
%
%
%
\begin{theorem} Soundness: If $T$ is a closed derivation tree with root node $C$, then $C$ is unsatisfiable in \INT.
\label{thm:soundness}
\end{theorem}


%
\begin{theorem} 
\label{thm:termination}
 Termination: Every derivation starting from a finite configuration $C$ where every variable is bounded is finite.
\end{theorem}
\subsection{Refutation Algorithm}
\label{subsec:alg}
\begin{algorithm}[t]
        \caption{Overview of the Refutation Procedure}
        \label{alg:overview}
        \KwIn{A set of constraints $C$}
        \KwOut{\unsat/\unknown}
        \Fn{{\sc Refute}$(C)$}{
        \DontPrintSemicolon
        $(B, \IntEqs, \IntNEqs \hdots \REqsN, \RNEqsN) \leftarrow \Split(C)$ \tcp*[l]{Figure~\ref{fig:split}} 
        \While{$\LiftLower((B, \IntEqs, \IntNEqs \hdots \REqsN, \RNEqsN))$\tcp*[l]{Section~\ref{subsec:calculus}-\ref{subsec:ilp}}} {\If{$\mathtt{CheckUnsat}((B, \IntEqs, \IntNEqs \hdots \REqsN, \RNEqsN))$\tcp*[l]{Section~\ref{subsec:calculus}}}{\KwRet{\unsat}}}
        \uIf{$S \leftarrow \Branch((B, \IntEqs, \IntNEqs \hdots \REqsN, \RNEqsN))$ 
        }
        {\For{$c_i \in S$}{\uIf{{\sc Refute}$(c_i) == \unknown$}{\KwRet{\unknown}}} \KwRet{\unsat}}
        {\KwRet {\unknown}}
        }
\end{algorithm}

Our implementation of the calculus follows the following strategy.  It first performs lifting via the \LiftEq and \LiftDiseq rules to exhaustion (i.e., until these rules no longer apply) and then attempts to refute each subsystem using the \UnsatOne and \UnsatDiseq rules for each \REqsN, $n \neq \infty$. If this fails, it applies the bounds-related rules  (\ConstrBds, \InfEq, \UnsatBds) and lowering rules (\LowerEq and \LowerDiseq), also to exhaustion. It then once again attempts refutation, but with $n=\infty$. If that also fails, it applies a branching rule, if applicable (trying first \PrimeSqr, then \RngLift), and then the entire strategy repeats.

Algorithm~\ref{alg:overview} presents an algorithmic view of the overall process.
Given constraints $C$, the procedure starts by encoding $C$ into a configuration using the $\Split$ routine, which creates the tuple encoding described in Section~\ref{subsec:overview} and applies the rules from Figure~\ref{fig:split}.
%
The main refutation procedure consists of the while loop in lines 3--7 and works as follows. 
First, it exchanges expressions between the different subsystems via lifting and lowering and tightens bounds or learns new equations from bounds (line 3).
Since lifting requires finding additional polynomials that are in the ideal, the implementation of the \LiftLower procedure utilizes the methods described in Sections~\ref{subsec:gb} and ~\ref{subsec:ilp} to find additional polynomials corresponding to exchangeable expressions.
The methods can be used individually or in combination, and we defer the discussion of their ordering to Section~\ref{sec:experiments}.
Next, at line 5, the algorithm invokes \CheckUnsat to attempt to refute any of the subsystems.
If \CheckUnsat returns true, the algorithm terminates with \unsat.
Otherwise,  the algorithm attempts to apply a branching rule (lines 8-12).  If successful, it recursively calls the {\sc Refute} procedure on a new set of constraints and returns \unsat if \emph{all} recursive calls perform successful refutation. Since our algorithm is intended only for refutation rather than model construction, it returns \unknown in all other cases. 
%
%
%

\subsection{Finding Liftable Equalities via Weighted \gbs}
\label{subsec:gb}
Recall that several rules in our refutation calculus require finding liftable or lowerable constraints that correspond to a polynomial in
the ideal of some existing set of polynomials.
As explained earlier, finding \emph{additional} liftable/lowerable constraints
is useful because they can make it easier to prove
unsatisfiability.
However, since ideals are
infinite, we need some algorithm for selecting which polynomials to target.
%
In this section, we present a method for computing ``useful'' polynomials---those that are likely to correspond to liftable constraints. We focus on the liftability of equalities only, because
all equalities are already lowerable, and disequalities do not form an ideal and thus are easy to enumerate.

Based on Lemma~\ref{lem:eq_lift}, a polynomial's bound endpoints are the most important indicators of liftability. 
As a result, we target  polynomials that we call \emph{near-zero polynomials}, or polynomials with a bound whose endpoints have small absolute value.  
Based on this definition, a polynomial with a bound $[-5,5]$ is \emph{closer to zero} than a polynomial with a bound of  $[-50,50]$ or $[49,50]$.

Since we cannot enumerate the ideal, one possibility is to
construct a basis for the ideal that contains polynomials that are more likely to refer to liftable constraints.
As discussed in Section~\ref{sec:bg}, one possible basis
is a \gb, which uses a monomial order.
A \gb is computed via an algorithm that often eliminates larger monomials with respect to the order.
Our idea is to compute a \gb using a carefully chosen monomial order
in which monomials that are closer to zero are smaller than those farther from zero.
The insight is that polynomials with near-zero monomials
are generally (though not always) more likely to be liftable.
To achieve this, our monomial order must encode information about bounds on the monomials.
To this end, we first take the maximum of the absolute values of the lower and upper bounds on each variable.
%
%
Additionally, since weighted monomial orders are determined by the dot product
of variable exponents and assigned weights, we apply logarithmic scaling to
ensure that the influence of a bound is appropriately adjusted by each variable's
degree.
Based on this intuition, we use a \emph{weighted reverse lexicographic}
monomial order (Sec.~\ref{sec:bg}),
with weights:

\footnotetext{The $\epsilon$ serves to avoid $\log(0)$ which is undefined.}

\begin{equation}
  \mathtt{weights} \leftarrow
  \left[
    \log \left(
      \max
      \left(
        \lvert B(\LVar_i)_1\rvert,
        \lvert B(\LVar_i)_2\rvert
      \right)
      + \epsilon\footnotemark
    \right)\,
    \mid \LVar_i \in \LVars
  \right]
\label{eq:weights}
\end{equation}
When applying the \LiftEq rule in the refutation calculus
our algorithm computes a \gb according to the order defined by
(\ref{eq:weights})
and only checks the liftability of the polynomials in this \gb.

Ideally, when we apply the \LiftEq  rule, we want to find a
\emph{complete} set of liftable equations---that is,
any other liftable constraint should be implied by the ones inferred by our
technique.
Our weighted \gb method does not have this completeness guarantee in general,
but the following theorem states a weaker completeness guarantee that it \emph{does} provide:

\begin{theorem}
  Let $C$ be a constraint system containing a modulus-$n$ equality subsystem
  \REqsN and  variable bounds $B$ s.t. for all $\LVar_i \in \LVars$, $0 \in B(\LVar_i)$.
  Let $ G = GB_{n,\le}(\dbk{\REqsN})$ be a \gb computed with a weighted reverse
  lexicographical order with weights from Equation~\ref{eq:weights}. 
  Finding liftable equalities derived from $I(\REqsN)$ by computing $G$ is complete if every generator $\dbk{e}\in G$ that has a leading monomial $\dbk{m}$ s.t. $\CalcBds(B,m)\subseteq[1-n,n-1]$ corresponds to a liftable expression $e \bmod n$ in $C$. 
  %


  \label{thm:lift:gb:partial_complete}
\end{theorem}
   Completeness implies that the ideal generated by the generators of $G$, which correspond to liftable expressions, contains all polynomials corresponding to liftable expressions from $I_n(\dbk{\REqsN})$. 

Intuitively, completeness means
  $ \{\dbk{p} : \CalcBds(B,p) \subseteq [1-n,n-1] \land \dbk{p} \in I_n(\dbk{\REqsN}) \} \subseteq I_n(\{\dbk{g} : \CalcBds(B,g) \subseteq[1-n,n-1] \land \dbk{g} \in G\})$.
A detailed proof can be found in Appendix~\ref{ap:proof:gb}.
This theorem is useful because it allows us to identify cases when additional heuristics for identifying liftable polynomials might be useful.

\subsection{Finding Liftable Equalities via Integer Linear Constraints}
\label{subsec:ilp}

Now we present another method, complementary to the one in Section~\ref{subsec:gb},
for finding liftable constraints.
This new method is motivated by the following observation:
In many cases, liftable equalities can be obtained from \emph{linear combinations} of existing  expressions---i.e., they are of the form $e' = \sum_{i=1}^l a_i e_i $ for existing expressions $e_i$ and constant coefficients $a_i$.
The method proposed here aims to find these unknown coefficients $a_1, \ldots, a_n$ by setting up an auxiliary \emph{integer linear constraint}. If this linear constraint  is feasible, then the inferred equality is guaranteed to satisfy the conditions from Lemma~\ref{lem:eq_lift}.

To emphasize that the coefficients $a_i$ are the unknowns being solved for, we depict them using a bold font below.
Our goal is to encode the condition $\CalcBds(B,\pnew) \subset [1-n,n-1]$ in the generated constraint.
Since \CalcBds depends on \pnew's coefficients,
we express those coefficients in terms of the $\mathbf{a}_i$.
Let $\dbk{m_1}, \dots, \dbk{m_t}$
be all the monomials present in \dbk{\REqsN},
including the constant monomial $1 \in \ZZ[X]$.
Moreover, let $\dbk{e_i}$ have coefficients $c_{i,j}$,
such that $\dbk{e_i} = \sum_{j=1}^t c_{i,j}\dbk{m_j}$.
Then,
the coefficient of $\dbk{m_j}$ in $\dbk{e'}$,
denoted $\coef_j$,
is the linear polynomial,
\( \coef_j \triangleq \sum_{i=1}^t \textbf{a}_i c_{i,j} \).

We can now express the relevant bounds.
The lower and upper bounds on each \textit{monomial}
$m_j$ can be concretely computed
as
\(u_j \triangleq \CalcBds(B,m_j)_1\)
and
\(l_j \triangleq \CalcBds(B,m_j)_0\).
Then, the lower and upper bounds on $\pnew$ itself
are respectively:
\begin{align*}
  \lowerbnd &\triangleq
  \textstyle\sum_{j=1}^{t} (
    l_j \cdot \relu(\coef_j)
    -
    u_j \cdot \relu(-\coef_j)
  )
  \\
  \upperbnd &\triangleq
  \textstyle\sum_{j=1}^{t} (
    u_j \cdot \relu(\coef_j)
    -
    l_j \cdot \relu(-\coef_j)
  ),
\end{align*}

\noindent
where $\relu(a)$ is $\max(a,0)$. \relu constraints capture how negating a monomial affects its upper and lower bounds: if \(\coef_j > 0\), \(u_j\) contributes to the upper bound of \(\pnew\); otherwise, it contributes to the lower bound. These constraints can be encoded using binary variables when the upper and lower bounds are known \cite{tsay2021partition}. Since the operations are performed modulo \(n\), all coefficients must lie within the interval \([1-n, n-1]\).
Our  encoding $\Phi$ is then:
\begin{align*}
  \Phi \triangleq \quad \lowerbnd > -n \quad\land\quad \upperbnd < n
  \quad\land\quad\textstyle\sum_{i=1}^k\left(\relu(\coef_j) + \relu(-\coef_j)\right) > 0
\end{align*}
The first two conditions ensure  \pnew is  liftable,
and the last ensures that $\pnew \neq 0$.

\begin{theorem}
  Let $C$ be a constraint system containing modulus-$n$ equality subsystem \REqsN and variable bounds $B$. $\Phi$ is satisfiable iff there exists a linear combination of the form 
  $\pnew = \sum_{i=1}^l a_i e_i $, $a_i \in \ZZ $ and $e_i\in \REqsN$,
   s.t. \pnew is liftable in $C$ and $\pnew \neq 0$.
\label{thm:lift:ilp:partial_complete}
\end{theorem}

In the case that $\Phi$ is satisfiable, we seek to find all linearly independent solutions, as different solutions correspond to different liftable equations, all of which may be useful for proving unsatisfiability. Thus, we add constraints ruling out linear combinations of solutions found so far and iterate until $\phi$ is unsatisfiable.

\section{Implementation}
\label{sec:implementation}
We implemented our refutation procedure in the cvc5 SMT solver \cite{cvc5} as an alternative to the existing nonlinear integer solver.
%
%
We use Singular v4.4.0 \cite{DGPS} for the algebraic components of our procedure and glpk \cite{glpk} v4.6.0 for solving the integer linear constraints problems from Section~\ref{subsec:ilp}.
%
%
Below we describe the key details in our implementation.


 \paragraph{Rewrites} Our implementation  disables all rewrites for the $\bmod$ operator but retains the remaining cvc5 rewrites and pre-processing passes. 
 \paragraph{Ideal Membership Check}
    To check if a polynomial version of an expression is in $I_{n}(\dbk{\REqsN})$
    we compute a \gb (GB) of \dbk{\REqsN} using \emph{graded reverse lexicographic order}   unless \dbk{\REqsN} is already a GB.
    (Recall that reduction by a GB is a complete ideal membership test.)
    To prevent GB computation from becoming a bottleneck, we impose a 30-second timeout.
    If a timeout occurs, we only check inclusion in the set instead of the ideal.
  \paragraph{Liftable Equalities via Integer Linear Constraints}
    Experimental results showed that integer linear constraint solvers  perform poorly on problems with large integer constants.  
    To address this issue, our implementation  uses an approximate version where we scale constants using signed log defined as $\slog_2(a) = \frac{|a|}{a}\log_2(a)$ and use a 30-second timeout. Details of our approximation, including empirical results that motivate this approximation can be found in Appendix~\ref{sec:aprox_imp}. All experimental results reported in the paper use this relaxation of the  encoding. 
    


\section{Experiments}%
\label{sec:experiments}
Our experiments are designed to answer two key empirical questions
\begin{enumerate*}
  \item How does our refutation procedure compare to the state of the art?
  \item Do our lifting algorithms (Sec.~\ref{subsec:gb} and~\ref{subsec:ilp}) improve performance?
\end{enumerate*}
All of our expriments are run on a cluster with Intel Xeon E5-2637 v4 CPUs. Each run is limited to one physical core, 8GB memory, and 20 minutes.
\subsection{Benchmarks}%
\label{sec:benchmarks}
The benchmarks used in our experimental evaluation consist of logical formulas encoding the correctness of simulating arithmetic operations in a given \textit{specification} domain using arithmetic in a corresponding \textit{base} domain. Each domain is either a finite field (\FFp or \FFq) or a bit-vector domain (\BV). Implementations that represent values in the specification domain as multiple \textit{limbs} in the base domain are referred to as \textit{multiprecision} (\texttt{m}), while those using a single representation are classified as \textit{single precision} (\texttt{s}).
%
%
%
\begin{table}[t]
  \centering
  \renewcommand{\tabcolsep}{1em}
  \begin{tabular}{llll}
    \toprule
    Family & Spec. & Base  & Reference \\
    \midrule
    \benchFfBvSp & $\FFp$ & $\BV$  & Montgomery arithmetic~\cite{montgomery1985modular}\\
    \benchFfBvMp & $\FFp$ & $\BV$  & Montgomery arithmetic~\cite{montgomery1985modular}\\
    \benchFfFfSp & $\FFp$ & $\FFq$   & Succinct labs~\cite{succinct2024goldilocks,goldilocks2015hamburg}\\
    \benchFfFfMp & $\FFp$ & $\FFq$  & o1-labs~\cite{o1labs2023nonnative}\\
    \benchBvFfMp & $\BV$ & $\FFp$  & xjSnark~\cite{kosba2018xjsnark}\\
    \bottomrule
  \end{tabular}%
  \tabCaptionSpace%
  \caption{
    Our benchmarks,
    which verify arithmetic implementations with
    various
    specification domains,
    base domains and numeric precisions.
  }
  \label{tab:benchmarks}
  \vspace{-0.2in}
\end{table}

As summarized in Table~\ref{tab:benchmarks},  our benchmarks fall under the following categories:
%
%
%
%
\begin{itemize}[leftmargin=*]
\item
  \benchFfBvSp 
  encode the correctness of single-precision
  \textit{Montgomery arithmetic}~\cite{montgomery1985modular},
  which implements arithmetic modulo $p$
  using bit-vector operations,
  without mod-$p$ reductions, and is common in prime field CPU implementations.
  The benchmarks verify the correctness of the \texttt{REDC} subroutine
  and the end-to-end correctness of Montgomery's approach for
  evaluating  expressions mod $p$.
  %
\item
  \benchFfBvMp encode correctness of multi-precision Montgomery arithmetic.
\item
  \benchFfFfSp encode implementations
  of arithmetic modulo the 31-bit Goldilocks prime~\cite{goldilocks2015hamburg}
  modulo a 255-bit prime (the order of the BLS 12-381 elliptic curve~\cite{bls12381,barreto2003constructing}).
  The benchmarks model a Succinct labs
  implementation~\cite{succinct2024goldilocks}
  that is used for recursive zero-knowledge proofs (ZKPs).
\item
  \benchFfFfMp encode the correctness of implementations
  of arithmetic modulo one 255-bit prime
  using arithmetic modulo another 255-bit prime.
  Multiple limbs are used
  to avoid unintended overflow in the base domain.
  The benchmarks model an
  o1-labs implementation~\cite{o1labs2023nonnative}
 used for recursive ZKPs.
\item
  \benchBvFfMp encode correctness of multi-precision implementations of bit-vector arithmetic
  modulo a 255-bit prime.
  The benchmarks
  model techniques from the xjSnark ZKP
  compiler~\cite{kosba2018xjsnark}
  that are used to check RSA signatures.
\end{itemize}

\paragraph{Determinism}
We include \textit{determinism}~\cite{CAV:OPBFBD24}
benchmarks in addition to correctness
for families with a finite field base domain.
%
Determinism is a weaker property than correctness
but 
 rules out most bugs in practice~\cite{chaliasos2024sok}. %
%
%
%
We do not include determinism \benchFfFfMp benchmarks
because they are correct but \textbf{non}deterministic.
 Additional benchmark statistics are included in Appendix~\ref{ap:bench_stats}.

\subsection{State of the Art Comparison}
\paragraph{Baselines}
Our benchmarks are SMT problems in \qfmia (Sec.~\ref{sec:fragment}).
Since all variables have finite bounds,
our benchmarks can also be encoded in \qfbv.
Our determinism benchmarks contain only
variable bounds and
equations modulo one prime
and can thus be encoded in \qfff
using standard \FF encodings of range constraints~\cite{kosba2018xjsnark}.
As a result, the baselines for our work
are existing solvers
for \qfnia, \qfbv, and \qfff.
We compare against
Yices v2.6.5~\cite{yices22},
cvc5 v1.2.2~\cite{cvc5},
bitwuzla v0.5.0~\cite{niemetz2023bitwuzla},
and
z3 v4.13.1~\cite{z3},
which are state of the art for these logics.
We evaluate bitwuzla with and without abstraction~\cite{niemetz2024scalable}
and cvc5's \qfff solver with and without split \gbs~\cite{CAV:OPBFBD24}.
%



\begin{table}[t]
\centering
\setlength{\tabcolsep}{3pt}
\begin{tabular}{lccccccccccc}
  \toprule
  \multicolumn{1}{c}{} & \multicolumn{6}{c}{Family} &  \multicolumn{1}{c}{} \\
  \cmidrule(rl){2-7} 
  \multicolumn{1}{c}{}  & \benchFfBvSp & \benchFfBvMp & \multicolumn{2}{c}{\benchFfFfSp} & \benchFfFfMp & \multicolumn{2}{c}{\benchBvFfMp} & \multicolumn{1}{c}{}\\
  \cmidrule(rl){2-2} \cmidrule(rl){3-3} \cmidrule(rl){4-5} \cmidrule(rl){6-6} \cmidrule(rl){7-8}
  Solver/Logic  &  cor & cor & cor & det & cor & cor & det  & Total  \\
  \cmidrule(r){1-1} \cmidrule(rl){2-8} \cmidrule(l){9-9} 
  \textbf{ours {\tiny QF\_MIA }} & \textbf{53} &  14 & \textbf{10} & \textbf{84} & \textbf{50} & \textbf{24} & 20 & \textbf{255} \\
  z3 {\tiny QF\_NIA }        & 49  & 21      & 3          & 34 & 40    &      22 & 21     & 190          \\ 
  cvc5  {\tiny QF\_NIA }   &  40  & 35 & 0        & 19      & \textbf{50}        & 14 & 8    &   166       \\ 
  yices  {\tiny QF\_NIA }   &  32  & 24 & 0         & 2     & 10 & 13    &     8      &   89       \\ 
    bitwuzla{\footnotesize (abst.)} {\tiny QF\_BV } & 49 &    \textbf{49}         & 0           & 2 & 25    &     9 & 8      & 142          \\
      bitwuzla {\tiny QF\_BV }   & 49   &    \textbf{49}     & 0          & 0 & 0    &     9 & 8      & 115          \\
       cvc5 {\tiny QF\_BV }       & 43   &  48    & 0          & 0 & 0    &     9 & 7      & 107           \\
  z3 {\tiny QF\_BV }          & 48   &  37  & 0          & 0 & 0    &     9 & 8     & 102           \\

  cvc5 {\tiny QF\_FF }   & -- & --  & -- & 1                      & -- & --          &  \textbf{23}   &   24           \\
     yices {\tiny QF\_FF }   & -- & --  & -- & 0                      & -- & --          &  20   &   20          \\
  cvc5   {\footnotesize(split GB) }{\tiny QF\_FF }   & --  & -- & -- & 2                   & --        &  --  &  9   &    11        \\

  \cmidrule(r){1-1} \cmidrule(rl){2-9}
  \# Benchmarks   & 72 &58 & 10 & 98 & 50 & 24 & 24 &         336            \\
  \bottomrule
  \end{tabular}
  \tabCaptionSpace
  \caption{Solved benchmarks. Here, cor stands for correctness and det for determinism. -- indicates the benchmark family cannot be encoded into the logic.}
  \label{tab:results}
  \vspace{-0.2in}
\end{table}

\begin{table}[t]
\centering
\setlength{\tabcolsep}{3pt}
\begin{tabular}{lcccccccccccc}
        \toprule
        \multicolumn{1}{c}{} & \multicolumn{6}{c}{Family} &  \multicolumn{1}{c}{} \\
        \cmidrule(rl){2-8} 
        \multicolumn{1}{c}{}  & \benchFfBvSp & \benchFfBvMp & \multicolumn{2}{c}{\benchFfFfSp} & \benchFfFfMp & \multicolumn{2}{c}{\benchBvFfMp} & \multicolumn{1}{c}{}\\
        \cmidrule(rl){2-2} \cmidrule(rl){3-3} \cmidrule(rl){4-5} \cmidrule(rl){6-6} \cmidrule(rl){7-8}
        Ablation  &  cor & cor & cor & det & cor & cor & det  & Total &  Uniq.   \\
        \cmidrule(r){1-1} \cmidrule(rl){2-10} 
        \cmidrule(rl){9-10} 
        \textbf{Weighted GB} & 53 &  14 & \textbf{10} & \textbf{84} & \textbf{50} & \textbf{24} & 20  & \textbf{255} & 11 \\
         Unweighted GB & 41 & \textbf{25} & \textbf{10} & 61 & \textbf{50} & 17   & 20 & 223 & 1 \\ 
        Lin. Constraints  & 37 & 18 & 0 & 7 & 0 & 7  & 3 & 72 & 3 \\
        Weighted GB \& Lin. & \textbf{54} &  23 & \textbf{10} & 73 & \textbf{50} & \textbf{24} & 20 & 254 & \textbf{12}  \\
           \cmidrule(r){1-1}
           \cmidrule(rl){2-10}
           \# Benchmarks  & 72 &58 & 10 & 98 & 50 & 24 & 24 &         336            \\
         \bottomrule

    \end{tabular}
    \tabCaptionSpace
  \tabCaptionSpace
  \caption{Results with different lifting algorithms}
  \label{tab:ablation}
\end{table} 

\begin{figure}[t]
\vspace{-0.1in}
  \centering
\includegraphics[width=0.9\textwidth]{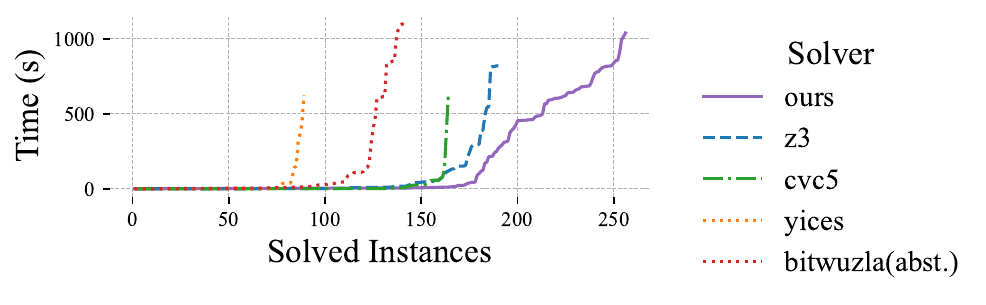}
\vspace{-0.15in}
    \caption{Benchmarks solved over time for top 5 solvers: ours (\OurLogic), z3 (\qfnia), cvc5 (\qfnia), yices (\qfnia), and bitwuzla w/ abstractions (\qfbv)}
\label{fig:cactus}
\vspace{-0.2in}
\end{figure}
\paragraph{Comparison}Table~\ref{tab:results} shows the number of \unsat benchmarks solved by family, type, logic, and tool, and Figure~\ref{fig:cactus} shows the performance for the top 5 solvers. 
Here,
`ours' stands for the best version of our solver: one with a lifting algorithm based on a weighted \gb (Section~\ref{subsec:gb}). 
This configuration outperforms existing tools on 5 out of the 7 categories, as well as in total benchmarks solved. 
It is also the most efficient and has the most unique solves, with 76 benchmarks not solved by any other solver. 
Of the 81 benchmarks our tool fails to solve, 27 remain unsolved by any solver. 
For these 81 benchmarks, 10 run out of memory, 29 time out, and 42 return \unknown.
Our solver performs the worst on the \benchFfBvMp family; it accounts for 34 of the \unknown benchmarks. 
We believe the poor performance in this domain is due to our algorithm's inability to find liftable equalities necessary to detect unsatisfiability, even though such equalities exist in the ideal. 
%
%


%
%
%

\subsection{Evaluation of Lifting Methods}
\label{subsec:lift_eval}

Recall that our method relies on identifying liftable equations, with Sections~\ref{subsec:gb} and~\ref{subsec:ilp} introducing two  potentially complementary approaches for this purpose. We now present the results of an evaluation comparing these lifting methods. As summarized in Table~\ref{tab:ablation}, the weighted \gb method achieves the best overall performance\textit{}. However, on the \benchFfBvMp benchmark family, the unweighted \gb method outperforms the weighted variant, supporting our hypothesis that the weighted \gb's inability to find certain liftable equalities contributes to its weaker performance in this category.\textit{}
Lifting based on integer linear constraints performs the poorest, identifying almost no liftable equalities required for refutation. Nevertheless, when combined with the weighted \gb method (in cases where the \gb method is not guaranteed to be complete  per Theorem~\ref{thm:lift:gb:partial_complete}), this hybrid approach has the most unique solves (12).



\section{Related work}%
\label{sec:relwork}

There are many SMT theory solvers for different kinds of modular integer
arithmetic.
Some solvers target theories of 
non-linear integer
arithmetic with an explicit modulus operator,
such as the theory of integers
\unskip~\cite{abraham2021deciding,caviness2012quantifier,cimatti2018incremental,marechal2016polyhedral,franzle2006efficient,tung2016rasat,weispfenning1997quantifier,jovanovic2013solving,moura2013model,jovanovic2011cutting,jovanovic2017solving,corzilius2015smt,bjorner2024arithmetic},
and the theory of bit-vectors
\unskip~\cite{%
  niemetz2024scalable,
  niemetz2020ternary,
  graham2020solving,
  brummayer2009boolector,
  niemetz2019towards}.
These solvers can be used to solve multimodular systems,
but generally perform poorly because they are designed to support
\textit{general} modular reduction
(i.e., reduction modulo a variable)
rather than reduction by constants.
There are solvers that efficiently support reduction by a single
constant or \textit{class} of constants.
For example,
finite field solvers
\unskip~\cite{%
LPAR:HKK23,
hadermsthesis,
hadersmt22,
CAV:OKTB23,
CAV:OPBFBD24}
support reduction modulo primes,
while bit-vector solvers support reduction modulo powers of two.
But, none of these can reason simultaneously about equations modulo a large
prime and powers of two.
The Omega test reasons about equations
modulo arbitrary constants---but the equations must be linear~\cite{pugh1991omega}.
Our procedure supports non-linear equations modulo arbitrary constants.

Some of the above solvers use \gbs: an algebraic tool
for understanding polynomial systems~\cite{buchberger1965algorithmus}.
For example, most similar to our refutation procedure, cvc5's finite field solver~\cite{CAV:OKTB23,CAV:OPBFBD24} also relies on a \gb to detect unsatisfiability of constraint subsystems. However, unlike our approach, it does not separate bounds or apply lifting, lowering, or weighted ordering techniques to share expressions across subsystems.

Many computer algebra systems (CASs)
implement algorithms for computing different kinds of \gbs
\unskip~\cite{bosma1997magma,
cocoalib,
zimmermann2018computational,
meurer2017sympy,
gap,
heck1993introduction,
DGPS,
eisenbud2001computations,
eder2021efficient,
davenport1992axiom,
wolfram1991mathematica}.
Following in this vein, our procedure uses
strong \gbs ~\cite{norton2001strong},
and our implementation uses the Singular CAS
\unskip~\cite{DGPS}.

Formal methods for modular arithmetic
have a long history of applications to cryptography.
Interactive theorem provers (ITPs) have been used to verify many cryptographic
implementations
\unskip~\cite{%
schwabe2021coq,
philipoom2018correct,
erbsen2020simple},
including the Fiat cryptography library,
which is used in all major web-browsers
\unskip~\cite{erbsen2018systematic}.
ITPs have also been used to verify zero-knowledge proofs (ZKPs)
\unskip~\cite{leo,leover,coda,fournet2016certified}.
But ITP-based verification requires significant manual effort and expertise.
SMT solvers for finite fields
\unskip~\cite{%
PLDI:PCWRVMCGFD23,
CAV:OWBB23}
and static analyses
\unskip~\cite{%
circomspect,
ecne,
zkap}
have been used to verify ZKPs
automatically, but these tools are limited by challenging
tradeoffs between
scalability and generality.
%

\section{Conclusion and Future Work}%
\label{sec:discuss}
In this paper we presented a novel refutation procedure for multimodular constraints and two algorithms for sharing lemmas: one based on a weighted \gb and another utilizing integer linear constraints. Our experiments demonstrate improvement over state-of-the-art solvers on benchmarks that arise from verifying cryptographic implementations. They also show the promise of lifting algorithms, both individually and in combination. Nevertheless, substantial future work remains.

First, as discussed in Sections~\ref{subsec:gb} and~\ref{subsec:ilp}, none of our lifting algorithms are complete. Exploring solutions with more completeness guarantees could boost the performance of our method. Second, as shown in Appendix~\ref{sec:aprox_imp}, lifting using linear integer constraints shows poor performance due to the limitations of existing solvers. A custom solver tailored to our encoding could lead to better results.
Third, our method focuses on unsatisfiable benchmarks and does not address model construction for satisfiable instances. However, we believe that leveraging the subsystem structure could also be beneficial for restricting the search space of possible assignments for satisfiable multimodular problems. 
Finally, our method treats \gb computation as a black box. Exploring more iterative methods based on s-polynomials could also boost performance.

\paragraph{Acknowledgements}
  We thank
  Ben Sepanski
  and
  Kostas Ferles
  for helpful conversations.
  We acknowledge funding from
  NSF grant number 2110397,
  the Stanford Center for Automated Reasoning,
  and the Simons foundation.

\vspace{1em}

{\fontsize{9pt}{11pt}\selectfont
\paragraph{Disclosure of Interest}Shankara Pailoor, Alp Bassa, and Işil Dillig are employed at Veridise.  
Alex Ozdemir and Sorawee Porncharoenwase previously worked at Veridise.  
Sorawee Porncharoenwase is currently employed at Amazon Web Services (AWS).  
Clark Barrett is an Amazon Scholar. 
}
\vspace*{-1em}
\begin{flushleft}
\footnotesize
\setlength{\parskip}{0pt}
\setlength{\itemsep}{0pt}
\bibliographystyle{abbrv}
\bibliography{refs}
\end{flushleft}

\appendix
\section{Proofs of lifting and lowering lemmas.}
\label{sec:lemmaproofs}
\setcounter{lemma}{0}
\begin{lemma}
Let $C$ be a multimodular system with bounds $B$ and modulus-$n$ equality subsystem $\REqsN$, with $n\in\ZZ^{+}$.
Then, an equality $e \bmod n = 0$ is liftable in $C$ if
(1) $\dbk{e}
\in I_n(\dbk{R_n^=})$
and
(2) $\CalcBds(B,e) \subseteq [1-n, n-1]$.
\end{lemma}
\begin{proof}
  Since $\dbk{e} \in I_n(\dbk{R_n^=})$, $e \bmod n = 0$ is implied by $C$.
  Thus, we simply must show that $C$ and $C \cup \{e = 0\}$ are equisatisfiable.  The right-to-left direction is trivial.  For the other direction, suppose $C$ is satisfiable.  Since
  $e \bmod n = 0$ is entailed by $C$ and since $\CalcBds(B,e)\subseteq [1-n, n-1]$, it follows that in any interpretation $\interp$ of $C$, $e^{\interp}$ (i.e., the value assigned by $\interp$ to $e$) must be divisible by $n$ and in the range $[1-n,n-1]$, meaning it must be equal to zero.  Thus, the same interpretation also satisfies $e = 0$.\qed
\end{proof}
\begin{lemma}
Let $C$ be a multimodular system with modulus-$n$ disequality subsystem $R_n^{\neq}$, with $n\in\ZZ^{+}$.
Then a disequality $e\bmod n \neq 0$ is liftable in $C$ w.r.t. $n$ if $e\in R_n^{\neq}$.
\end{lemma}
\begin{proof}
  Suppose $e \in \RNEqsN$.  This means that $C$ implies $e \bmod n \ne 0$, so we must simply show that $C$ and $C \cup \{e \neq 0\}$ are equisatisfiable.  The right-to-left direction is trivial.  For the other direction, suppose $C$ is satisfied by some interpretation $\interp$.  We know that $e^{\interp}$ cannot be divisible by n.  So in particular, it cannot be 0.  Thus $\interp$ also satisfies $e \neq 0$.\qed
\end{proof}
\begin{lemma}
  Let $C$ be a multimodular system with integer equalities $\IntEqs$.
  If $e$ is an expression and $\dbk{e} \in  I_{\infty}(\dbk{\IntEqs})$, then $e=0$ is lowerable with respect to every $n\in\ZZ^{+}$.
\end{lemma}
\begin{proof}
Since $\dbk{e} \in I_{\infty}(\dbk{\IntEqs})$, $e=0$ is implied by $C$.
  Thus, we simply must show that $C$ and $C \cup \{e \bmod n = 0\}$ are equisatisfiable.  The right-to-left direction is trivial.  For the other direction, suppose $C$ is satisfiable.  Since
  $e = 0$ is entailed by $C$, it follows that in any interpretation $\interp$ of $C$, $e^{\interp} = 0$.  Thus, the same interpretation also satisfies $e \bmod n = 0$ for every $n\in\ZZ^{+}$.\qed
\end{proof}
\begin{lemma}
    Let $C$ be a multimodular system with integer disequalities $\IntNEqs$ and bounds $B$.
    Then, if $n\in\ZZ^{+}$, a disequality $e \neq 0$ is lowerable in $C$ w.r.t to $n$ if (1) $e \in \IntNEqs$ and (2) $\CalcBds(B,e) \subseteq [1-n,n-1]$ 
\end{lemma}
\begin{proof}
Since $e \in \IntNEqs$, $e\ne 0$ is implied by $C$.
  Thus, we simply must show that $C$ and $C \cup \{e \bmod n \ne 0\}$ are equisatisfiable.  The right-to-left direction is trivial.  For the other direction, suppose $C$ is satisfiable.  Since $e \ne 0$ is entailed by $C$, in any interpretation $\interp$ of $C$, $e^{\interp} \ne 0$.  We also know that $e^\interp \in [1-n,n-1]$.  Thus, $e^\interp$ cannot be divisible by $n$.  Thus $\interp$ also satisfies $e \bmod n \ne 0$.\qed
\end{proof}
\label{ap:lemma_proofs}
\section{Proofs of soundness and termination of our calculus}
\setcounter{theorem}{0}
\begin{theorem} Soundness: If $T$ is a closed derivation tree with root node $C$, then $C$ is unsatisfiable in \INT.
\end{theorem}
\begin{proof}
The proof reduces to a local soundness argument for each of the calculus rules. 
 \OneGB and \UnsatDiseq are sound by properties of ideals \cite{david1991ideals}.
 (Recall $1\in I(f_1,\hdots, f_n)$ implies $f_1,\hdots, f_n$ have no common zeros.) 
 Soundness of \InconBds, \ConstrBds, and \InfEq follows from basic arithmetic reasoning about bounds and intervals~\cite{intervalarithmetic}. 
 As mentioned earlier, \LiftEq and  \LiftDiseq are sound due to Lemmas~\ref{lem:eq_lift} and ~\ref{lem:diseq_lift}.
 Since $\smod$ preserves modulus congruence classes (i.e., $(a \smod n )\bmod n \equiv a \bmod n$), the  $\simplify_n$ rewrite is sound in a  modulus-$n$ subsystem. 
In combination with lemmas~\ref{lem:eq_lower} and \ref{lem:diseq_lower}, this makes rules  \LowerEq, and \LowerDiseq also sound.  
\RngLift is sound because if $e \bmod n \approx 0$ holds, then for some fresh variable $\LVar$, $e \approx n\LVar$ must hold. If $\CalcBds(B,e) \subseteq (1-2n, 2n-1)$, then the possible value of $\LVar$ must be within the interval $[-1,1]$ and we can do a case split on the three possible values of $\LVar$.
\PrimeSqr is sound because if $n$ is prime, then we are in a field, and thus, there are no zero divisors (no nonzero terms that when multiplied are equal to zero).  Then, $ s(s-1) \bmod n \approx 0$ implies that either $s\bmod n \approx 0$ or $(s-1)\bmod n \approx 0$ holds. Since all rules are sound, it follows the procedure is sound.
%
\qed
\end{proof}

\begin{theorem} 
  Termination: Every derivation starting from a finite configuration $C$ where every variable is bounded is finite.
\end{theorem}
\begin{proof}

Let $T$ be some derivation tree whose root is a finite configuration $C$ where every variable is bounded. We will prove termination by showing that every possible branch of $T$ can be mapped to a strictly descending sequence in a well-founded order. 
We do this by showing how to construct a tuple from a configuration (other than \unsat) in such a way that every rule application reduces the tuple with respect to some well-founded order.  Note that no rules apply to the \unsat configuration, so any branch containing \unsat must end with the \unsat configuration.  Now, fix a configuration $C$ different from $\unsat$.
The first entry in the tuple keeps track of how tight the bounds are on all of the variables.  Let $\Bd$ count the total number of integers in the intervals mapped to by $B$: $\Bd = \sum_{i=1}^k \relu(B(\LVar_i)_1 - B(\LVar_i)_0 + 1)$.

Next are entries that track disequalities.
Since the set of disequalities do not form an ideal we will argue that new disequalities can only appear through lowering (when we apply the $\simplify_n$ rewrite). As there are only finitely many starting disequalities and finitely many moduli, there are only finitely many number of disequalities that can be created.
We formalize this as follows.
We first overload the definition of $\simplify_n$ to apply to sets of expressions. For some set of expressions $G$, $\simplify_n(G) = \{\simplify_n(e) \mid e \in G\}$. Also, note that applying $\simplify_{n_1}$ to $\simplify_{n_2}(e)$, where $n_1 \ge n_2$ has no effect.
We define a function that gathers all possible disequalities as $\mathcal{F}(A,S)$ where $A = \{a_1\hdots a_j\} \subset \ZZ^{+}$, with $a_i < a_{i+1}$ for $i\in[1,j-1]$ and $S$ is a finite set of expressions.  $\mathcal{F}$ applies every possible combination of $\simplify_a$ for subsets of $A$ to every expression in $S$.  We define $\mathcal{F}$ recursively as follows.
\begin{align*}
  \mathcal{F}(\emptyset, S) & = S\\
  \mathcal{F}(A, S) & = \simplify_{a_1}(\mathcal{F}(A\setminus\{a_1\},S)) \cup \mathcal{F}(A\setminus\{a_1\},S)
\end{align*}
Now, let $A=\{n \mid n\in\ZZNN \land \RNEqsN \in C\}$, let $A_{\infty} = A \cup \{\infty\}$, and let $\ND_n$ denote the number of possible disequalities \emph{not} in the current \RNEqsN: $ \ND_n = \lvert \mathcal{F}(A, \IntNEqs \cup \bigcup_{i \in A} R_{i}^{\not\approx} ) \setminus \RNEqsN \rvert $, where $n\in A_{\infty}$.  Following \Bd, the next $|A_{\infty}|$ entries in the tuple are $\ND_n$, $n\in A_{\infty}$.

Finally, we have entries for each ideal $I_n(\dbk{\REqsN})$, $n\in A_{\infty}$.
For these entries, we rely on the Noetherian ring ascending chain condition: if a ring is Noetherian, every increasing sequence of ideals in the ring has a largest element (both $\ZZ[\Vars]$ and $\ZZ_n[\Vars]$ are Noetherian). 

Thus, for a non-\unsat configuration $C$, we can construct the following tuple: $(\Bd, \ND_{a_1}\dots \ND_{a_j}, I_{a_1}(\dbk{\REqsN}) \dots, I_{a_j}(\dbk{\REqsN})))$, where $A_\infty = \{a_1,\dots,a_j\}$.  We define a total well-founded order $<_t$ on such tuples as the lexicographic order induced by:
\begin{itemize}[label=\textbullet]
    \item  the standard non-negative integer ordering $<$ for $\Bd$;
    \item  the standard non-negative integer ordering $<$ for each $ND_n$;
    \item  an order \OrderI for ideals defined as follows: $I \OrderI I'$ iff $I \supset I'$.  This is well-founded by the Noetherian ring ascending chain condition mentioned above.
\end{itemize}

We now prove termination
by showing that every rule application either decreases the tuple for the configuration according to $<_{t}$
or results in \unsat.
\UnsatOne, \UnsatDiseq, \InconBds all return \unsat.
\ConstrBds decreases \Bd by tightening the spread of bounds for some variable $\LVar_i$. 
\InfEq and \LiftEq expand $I_{\infty}(\dbk{\IntEqs})$. 
\LiftDiseq decreases $\mathit{ND}_{\infty}$.
\LowerEq expands $I_{n}(\dbk{\REqsN})$.
\LowerDiseq decreases $\mathit{ND}_n$.
Each conclusion of \RngLift expands $I_{\infty}(\dbk{\IntEqs})$. 
Each conclusion of \PrimeSqr expands $I_{n}(\dbk{\REqsN})$.
Thus, each rule decreases according to a well-founded order, and so the calculus is terminating.
\qed
\end{proof}

\label{ap:sound_term}
\section{Proof of partial completeness for lifting via a weighted \gb}
\label{ap:proof:gb}
To prove Theorem~\ref{thm:lift:gb:partial_complete}, we rely on two auxiliary lemmas. 
    \begin{lemma}
    Let $\sigma$ be a reverse lexicographical monomial order  with a weight function $w$ defined using weights in Equation~\ref{eq:weights}. For any two monomials $\dbk{m}$,$\dbk{s}$ if   $\dbk{m} \geq_{\sigma} \dbk{s}$ and  $\calcBds(B,m)\subseteq[1-n,n-1]$, then $\calcBds(B,s)\subseteq[1-n,n-1]$.
\label{lem:lift_preserv}
\end{lemma}
\begin{proof}
To prove the claim, we will first use induction on number of variables in a monomial to show that for any monomial $\dbk{m}$. $$w(\dbk{m}) =  \log(\max( \lvert\calcBds(B,m)_1\rvert, \lvert\calcBds(B,m)_2\rvert)) \footnotemark$$
\footnotetext{we omit $\epsilon$ for simplicity as this does not change the proof} Let $m = \prod_{i=0}^k\LVar_i^{e_i}$  for $\LVar_i \in \LVars, e_i \in \ZZ_{\geq0}$ and recall by Equation~\ref{eq:weights}, $w(\dbk{m}) = \sum_{i=0}^{k}e_i \cdot \log(\max(\lvert B(\LVar_i)_1 \rvert,\lvert B( \LVar_i)_2 \rvert )) $. Which by properties of logarithms \\ $ =\log(\prod_{i=0}^{k}\max(|B(, \LVar_i)_1|,|B( \LVar_i)_2| )^{e_i})$.
By interval arithmetic, for $k=2$, \\ 
$\calcBds(B,\LVar_1\LVar_2) = \Big( \min(B(\LVar_1)B(\LVar_2)_1,B(\LVar_1)_1B(\LVar_2)_2,B(\LVar_1)_2B(\LVar_2)_1,$\\ $B(\LVar_1)_2B(\LVar_2)_2), \max((B(\LVar_1)B(\LVar_2)_1,B(\LVar_1)_1B(\LVar_2)_2, B(\LVar_1)_2B(\LVar_2)_1B(\LVar_1)_2B(\LVar_2)_2))\Big)$.
It follows  $\Big(\max(|B(\LVar_1)_1|, |B(\LVar_1)_2|) \max(|B(\LVar_2)_1|, |B(\LVar_2)_2|) \Big) $ will be included in either the first or second entry of the $\calcBds(B,\LVar_1\LVar_2)$ tuple (depending on the signs of the lower and upper bounds of $\LVar_1$ and $\LVar_2$), and dominate the absolute value of the interval's endpoints. Thus $\max(|\calcBds(B, \LVar_1\LVar_2)_1| , |\calcBds(B, $ \\ $ \LVar_1\LVar_2)_2|) $  $ = \prod_{i=1}^{2}\max(|B(\LVar_i)_1|, |B(\LVar_i)_2|)  = w(x_1x_2) $. The inductive case is equivalent, but uses $\calcBds(B,\prod_{i}^{n}\LVar_i^{e_i})$ for the monomial $\prod_{i}^{n}x_i^{e_i}$ instead of the map $B$. The proof of the lemma becomes trivial; $\calcBds(B,m) \subseteq [1-n,n-1]$ means $w(\dbk{m}) \leq \log(n-1)$ and if  $\dbk{m} \geq_{\sigma} \dbk{s}$ then  $w(\dbk{s}) \leq \log(n-1)$ and $\calcBds(B,s) \subseteq [1-n,n-1]$.  \qed
\end{proof}

\begin{lemma}
\label{lem:lift_poly_to_mono}
   Fix $B$ a bounds map s.t. for all $
   \LVar_i \in \LVars$, $0 \in B(\LVar_i).$  If for a polynomial $\dbk{e}$, $\calcBds(B,e) \subseteq[1-n,n-1]$  then for all terms $\dbk{t}$ of $\dbk{e}$, $\calcBds(B, t) \subseteq[1-n,n-1]$ 
\end{lemma}
\begin{proof}
By set up, since $
   \LVar_i \in \LVars$, $0 \in B(\LVar_i)$, every $\LVar_i \in e$ can be zero and every $x_i \in \dbk{e}$ can be zero. Which means every monomial and thus every term of $\dbk{e}$ can be zero. It follows that the minimum possible value of any term must be $\leq 0$   and the maximum value must be $\geq 0$. Since a polynomial is the sum of its terms by interval arithmetic its lower/upper bound is the sum of the lower/upper bounds of its terms. The lower bound can only become more negative by addition (as all terms' lower bounds are $\leq 0$) and the upper bound can only become more positive by addition (as all terms' upper bounds are $\geq 0$). It follows that if there existed a term $\dbk{t} \in \dbk{e}$ s.t. $\CalcBds(t) \not \subset [1-n,n-1]$ then $\CalcBds(e) \not \subset [1-n,n-1]$. \qed

\end{proof}
We are now ready to prove Theorem~\ref{thm:lift:gb:partial_complete}. 
\begin{theorem}
  Let $C$ be a constraint system containing a modulus-$n$ equality subsystem
  \REqsN and  variable bounds $B$ s.t. for all $\LVar_i \in \LVars$, $0 \in B(\LVar_i)$.
  Let $ G = GB_{n,\le}(\dbk{\REqsN})$ be a \gb computed with a weighted reverse
  lexicographical order with weights from Equation~\ref{eq:weights}. 
  Finding liftable equalities derived from $I(\REqsN)$ by computing $G$ is complete if every generator $\dbk{e}\in G$, that has a leading monomial $\dbk{m}$ s.t. $\CalcBds(B,m)\subseteq[1-n,n-1]$, corresponds to a liftable expression $e \bmod n$ in $C$. 
\end{theorem}
\begin{proof}
Let e be some liftable expression in $C$ s.t $\dbk{e} \in I_n(\REqsN)$.  Let $H$ be the set of liftable expressions $e_i$ s.t $\dbk{e_i} \in G$. And let $\le_{\sigma}$ be the weighted reverse lexicographical order computed with a weight function $w$ based on weights from Equation~\ref{eq:weights}.

We will show that if the conditions of the theorem are observed then $\dbk{e} \in I_n(H)$.
  Since $\dbk{e}\in I_n(\REqsN)$, by properties of a \gb
  there exists a reduction sequence
  $\dbk{e} \rightarrow \dbk{e}  - h_1 g_1 \rightarrow  \dbk{e}  -  h_1 g_1  - h_2g_2 \rightarrow
  \cdots \rightarrow  \dbk{e}  -  h_1 g_1  - \cdots - h_lg_l = 0$,
  for $g_i \in G$.
  Let $p_i$ denote the intermediary polynomials formed;
  that is, $p_i = \dbk{e} - \sum_{j=1}^i h_jg_j$
  By the definition of reduction in PIRs, $\leadt(g_i)$ must divide some term of $p_{i-1}$. Let this term equal $a_{i-1}s_{i-1}$ where $a_{i-1} \in \ZZ$ is the coefficient and $s_{i-1}$ the monomial. Then for some monomial $m_i \neq 0$,
  $\lm(g_i) * m_i = s_{i-1}$.
  By Equation~\ref{eq:weights}, weights are non negative.
  $  w(\lm(g_i) * m_i)
    = w(\lm(g_i))*1 + w(m_i)*1
    = w(s_{i-1})$
    it follows
        $w(\lm(g_i)) \leq w(s_{i-1}) \leq w(\lm(p_{i-1})) $. 
   By induction
    on $i$,
   $w(\lm(g_i)) \leq w(\lm(\dbk{e}))$.
  Since $e$ is liftable, $\CalcBds(B,e)\subseteq [1-n,n-1]$ and by Lemma~\ref{lem:lift_poly_to_mono} for $\dbk{m_e} = \lm(\dbk{e})$,
  $ \CalcBds(B,m_e)\subseteq [1-n,n-1]$.
  From Lemma~\ref{lem:lift_preserv} it follows that for $\dbk{m_{g_i}} = \lm(g_i)$, $\CalcBds(B,m_{g_i})\subseteq [1-n,n-1]$.
Conditions of the theorem state that if the bounds of the leading monomial of a polynomial are within $[1-n,n-1]$ then the expression corresponding to the polynomial must be liftable. 
Thus all expressions corresponding to $g_i$ must be liftable and the original reduction sequence shows that
  $\dbk{e} \in I_n(H)$. 
  \qed
\end{proof}
\label{ap:proof:gb}
\section{Proof of partial completeness for lifting via integer linear constraints}
\label{ap:proof:ilp}
\begin{theorem}
  Let $C$ be a constraint system containing modulus-$n$ equality subsystem \REqsN and variable bounds $B$. $\Phi$ is satisfiable iff there exists a linear combination of the form 
  $\pnew = \sum_{i=1}^l a_i e_i $, $a_i \in \ZZ $ and $e_i\in \REqsN$
   s.t. \pnew is liftable in $C$ and $\pnew \neq 0$.
\end{theorem}
\begin{proof}
$\Rightarrow$ Assume $\Phi$ is satisfiable. Then there exists a solution vector of the form $a_1 \hdots a_l$, $a_i \in \ZZ$. And $ - n > \lowerbnd \triangleq
  \textstyle\sum_{j=1}^{t} (
    l_j \cdot \relu(\sum_{i=1}^t \textbf{a}_i c_{i,j})
    -
    u_j \cdot \relu(-\sum_{i=1}^t \textbf{a}_i c_{i,j})$. By definition $l_j$ and $u_j$ correspond to the upper and lower bounds of a monomial $m_j$. By setup of $\Phi$, for $e_i \in \REqsN$, $\dbk{e_i} = \sum_{j=1}^t c_{i,j}\dbk{m_j}$, and $\sum_{i=1}^t \textbf{a}_i c_{i,j}$ corresponds to the coefficient of $\dbk{m_j}$ in the polynomial $\dbk{e}' = \sum_{i=1}^l a_i \dbk{e_i}$. By linear arithmetic and properties of \relu (to compute the lower bound we add lower bounds if coefficient is positive and subtract upper bound if it is negative), $\lowerbnd = \CalcBds(B, e
')_o < -n$. The proof for the upper bound is symmetrical. Thus $\upperbnd = \CalcBds(B,e')_1 > n$. It follows  $\CalcBds(B,e') \subseteq [1-n,n-1]$, meaning $e'$ is liftable. Since  $\dbk{e}' = \sum_{i=1}^l a_i \dbk{e_i} $,it follows $e' = \sum_{i=1}^l a_i e_i$.  Finally, since $\Phi$ is satisfiable, $\sum_{i=1}^k\left(\relu(\sum_{i=1}^t \textbf{a}_i c_{i,j}) + \relu(-\sum_{i=1}^t \textbf{a}_i c_{i,j})\right) > 0$. By set up of $\Phi$ and properties of \relu there must exist at least one monomial $\dbk{m_i} \in \dbk{e'}$ with a nonzero coefficient and thus $\dbk{e'} \neq 0$ and $e' \neq 0$. 
\\

\noindent $\Leftarrow$ Assume there exists a  linear combination of the form 
  $\pnew = \sum_{i=1}^l a_i e_i $, $a_i \in \ZZ $ and $e_i\in \REqsN$ s.t $\pnew$ is liftable in $C$ and $\pnew \neq 0$.  Since $e'$ is liftable $\CalcBds(B,e') \subseteq [1-n,n-1]$. By set up of $\Phi$ and the explanation above $\lowerbnd$ and $\upperbnd$ in $\Phi$ correspond to $\CalcBds(B,e')_0$ and  $\CalcBds(B,e')_1$ respectively. Thus $\lowerbnd > -n$ and  $\upperbnd < n$ are satisfiable by the vector of coefficients $a_1 \hdots a_l$ that forms $e'$. Further since $e' \neq 0$, it follows that there must exist at least one $a_i \neq 0$. From set up of $\Phi$ it follows that $\sum_{i=1}^k\left(\relu(\sum_{i=1}^t \textbf{a}_i c_{i,j}) + \relu(-\sum_{i=1}^t \textbf{a}_i c_{i,j})\right) > 0$. Thus $\Phi$ is satisfiable. \qed
  \end{proof}
\section{Approximate implementation of lifting via integer linear constraints}
\label{sec:aprox_imp}
As mentioned in Section~\ref{sec:implementation}, instead of the full version of integer linear constraint encoding defined in Section~\ref{subsec:ilp}, we use an approximation. The approximate encoding scales all integer constants in the full version using signed log, defined as $\slog_2(a) = \frac{|a|}{a}\log_2(a)$, and scales the solution vector by dividing all entries by their greatest common divisor. The table below compares the results of lifting using the full integer linear constraint encoding and the approximation.
\begin{table}[H]
\centering
\setlength{\tabcolsep}{3pt}
\begin{tabular}{lcccccccccccc}
        \toprule
        \multicolumn{1}{c}{} & \multicolumn{6}{c}{Family} &  \multicolumn{1}{c}{} \\
        \cmidrule(rl){2-7} 
        \multicolumn{1}{c}{}  & \benchFfBvSp & \benchFfBvMp & \multicolumn{2}{c}{\benchFfFfSp} & \benchFfFfMp & \multicolumn{2}{c}{\benchBvFfMp} & \multicolumn{1}{c}{}\\
        \cmidrule(rl){2-2} \cmidrule(rl){3-3} \cmidrule(rl){4-5} \cmidrule(rl){6-6} \cmidrule(rl){7-8}
        Ablation  &  cor & cor & cor & det & cor & cor & det  & Total   \\
        \cmidrule(r){1-1} \cmidrule(rl){2-9} 
        Full Lin. Constraints & 33 & 0 & 0 & 7 & 0 & 0  & 0 & 40  \\
        
        Aprox. Lin. Constraints & 37 & 18 & 0 & 7 & 0 & 7  & 3 & 72  \\
           \cmidrule(r){1-1}
           \cmidrule(rl){2-9}
           \# Benchmarks   & 72 &58 & 10 & 98 & 50 & 24 & 24 &  336  \\
 \bottomrule
\end{tabular}
\end{table}
While both algorithms struggle to find liftable equalities, the approximate encoding performs slightly better and thus we use this relaxation of the encoding for the experimental results in Section~\ref{sec:experiments}. 

For both algorithms we rule out linear combinations of solutions found as for as follows:
Given a matrix $A = \{\overrightarrow{a}_1, \dots, \overrightarrow{a}_\ell\}$, $\overrightarrow{a_i}$ is a solution to $\Phi$  and all $\overrightarrow{a_i}$ are \QQ-independent we can modify $\Phi$ to find a $\pnew_{\ell+1}$ that is not a
 \QQ-linear combination of $A$.
We compute $H$ as the left null space of $A$ (a matrix in $\QQ^{k \times (\ell - k)}$ entries $h_{j,i}$ and columns that are orthogonal to $A$'s).
 We then normalize $H$ to be a $\ZZ$ matrix by multiplying by lest common multiple of all denominators in each column (thus preserving independence). 
%
%
%
%
Then, we compute $U$ as the left null space of $H$
(a matrix in $\QQ^{k \times (\ell - k)}$ entries $u_{j,i}$ and columns that are orthogonal to $H$'s).
We introduce integer variables $\mathbf{b}_1, \dots, \mathbf{b}_k$,
and define $\Phi'$ as:
\begin{align*}
  \Phi'
  \triangleq
  \Phi~&\land~
  \textstyle\bigwedge_{j=1}^{k}
    \Big(
      \coef_j = \textstyle\sum_{i=1}^{\ell} \mathbf{b}_ia_{j,i} + \textstyle\sum_{i=\ell+1}^{k}\mathbf{b}_ih_{j,i}
    \Big)
    \\
       &\land
  \Big(
    \textstyle\bigvee_{i=\ell+1}^{k}\mathbf{b}_i \neq 0
  \Big)
\end{align*}
 
 Since $\Phi'$ contains disjunctions, to solve it, we split each disjunct into its own integer linear constraint. For the approximate encoding we \emph{do not} scale $h_{j,i}$ or $a_{j,i}$ as $\log_2$ does not preserve independence and in practice both $h_{j,i}$ and $a_{j,i}$ tend to be fairly small.
\section{Additional Benchmark Statistics}
\begin{table}[h]
\centering
\setlength{\tabcolsep}{3pt}
\begin{tabular}{lccccccccccc}
        \toprule
        \multicolumn{1}{c}{} & \multicolumn{6}{c}{Family} &  \multicolumn{1}{c}{} \\
        \cmidrule(rl){2-7} 
        \multicolumn{1}{c}{}  & \benchFfBvSp & \benchFfBvMp & \multicolumn{2}{c}{\benchFfFfSp} & \benchFfFfMp & \multicolumn{2}{c}{\benchBvFfMp} & \multicolumn{1}{c}{}\\
        \cmidrule(rl){2-2} \cmidrule(rl){3-3} \cmidrule(rl){4-5} \cmidrule(rl){6-6} \cmidrule(rl){7-8}
          &  cor & cor & cor & det & cor & cor & det    \\
        \cmidrule(r){1-1} \cmidrule(rl){2-8} 
       
\# of equalities (min) &
8.0 &
43.0 &
2.0 &
24.0 &
9.0 &
3.0 &
8.0 \\
\# of equalities (mode) &
19.5 &
43.0  &
4.0 &
105.0 &
40.5 &
7.0 &
38.0 \\
\# of equalities (max)  &
45.0 &
6,403.0 &
5.0 &
186.0 &
72.0 &
22.0 &
52.0 \\
\# of disequalities (min)  &
0.0  &
0.0 & 
1.0 &
1.0 &
1.0 &
1.0  &
1.0 \\
\# of disequalities (mode)  &
0.5 &  
0.5 & 
1.0  &
1.0 &
1.0 &
1.0 &
1.0 \\
\# of disequalities (max)  &
1.0 &
1.0 & 
1.0 & 
1.0 &
1.0 &
1.0  &
1.0 \\
\# of inequalities (min)  &
8.0 &
52.0 &
15.0 &
52.0 &
12.0 & 
16.0 & 
32.0 \\
\# of inequalities (mode)  &
15.0 &
53.0 &
27.0 &
187.0 &
39.0 &
65.0 &
130.0 \\
\# of inequalities (max)  &
44.0 &
6,534 &
36.0 &
322.0 &
66.0 &
96.0 &
192.0\\
\# of variables (min)  &
9.0 &
49.0 &
5.0 &
20.0 &
9.0 &
6.0 &
12.0 \\
\# of variables (mode)  &
21.0 &
49.0 &
9.0 &
74.0 & 
36.0 & 
22.0 &
44.0 \\
\# of variables (max)  &
48.0 &
6,499.0 &
12.0 &
128.0 &
63.0 &
37.0 &
74.0 \\
\# of disjunctions (min)  &
0.0  &
0.0  &
0.0 &
8.0 & 
4.0 & 
0.0 &
0.0 \\
\# of disjunctions (mode)  &
0.5 &
0.5 &
0.0 & 
44.0 & 
22.0 &
0.0 &
0.0 \\
\# of disjunctions (max) & 
1.0 &
1.0 &
0.0 &
80.0 &
40.0 &
0.0 &
0.0 \\

           \cmidrule(r){1-1}
           \cmidrule(rl){2-8}
           \# of Benchmarks   & 72 &58 & 10 & 98 & 50 & 24 & 24  \\
 \bottomrule
\end{tabular}
\end{table}
\label{ap:bench_stats}
\newpage
\section{Additional Algorithm Statistics}
\begin{table}[h!]
\centering
\setlength{\tabcolsep}{3pt}
\begin{tabular}{lcccccccc}
        \toprule
        \multicolumn{1}{c}{} & \multicolumn{6}{c}{Family} &  \multicolumn{1}{c}{} \\
        \cmidrule(rl){2-7} 
        \multicolumn{1}{c}{}  & \benchFfBvSp & \benchFfBvMp & \multicolumn{2}{c}{\benchFfFfSp} & \benchFfFfMp & \multicolumn{2}{c}{\benchBvFfMp} & \multicolumn{1}{c}{}\\
        \cmidrule(rl){2-2} \cmidrule(rl){3-3} \cmidrule(rl){4-5} \cmidrule(rl){6-6} \cmidrule(rl){7-8}
          &  cor & cor & cor & det & cor & cor & det   \\
        \cmidrule(r){1-1} \cmidrule(rl){2-8} 
       \# of partitions (min) &
6.0 &
14.0 &
6.0 & 
6.0 &
16.0 &
6.0 &
6.0 \\
\# of partitions (mode) &
10.0 &
18.0 &
6.0 &
22.0 &
16.0 &
14.0 &
8.0 \\
\# of partitions (max) &
14.0 &
324.0 &
6.0 &
22.0 &
16.0 &
72.0 &
60.0 \\
\# of exp lift/lowered (min) &
6.0 &
0.0 &
7.0 &
17.0 &
24.0 &
4.0 &
20.0 \\
\# of exp lift/lowered (mode) &
48.0 &
0.0 &
11.0 &
32.0 &
60.0 &
8.0 &
24.0 \\
\# of exp lift/lowered (max) &
89,800.0  &
4,328.0 &
13.0 &
1,306.0 &
96.0 &
357.0 &
816.0 \\
avg. \# poly per GB  (min) &
2.5 &
1.0 &
2.0 &
6.1 &
2.0 &
2.7 &
5.5 \\
avg. \# poly per GB  (mode) &
6.0 &
13.6 &
2.0 &
158.2 &
2.0 &
4.9 &
11.6 \\
avg. \# poly per GB  (max) &
24.8 &
3,168.0 &
5.1 &
460.4 &
2.0 &
70.9 &
90.1 \\
 \# GB computes (min)  &
2.0 &
0.0 &
4.0 &
15.0 &
8.0 &
3.0 &
18.0 \\
 \# GB computes (mode) &
32.0  &
212.0 &
8.0 &
30.0 &
8.0 &
7.0 &
22.0 \\
 \# GB computes (max) &
67,386.0 &
4,271.0 &
10.0 &
1,313.0 &
8.0 &
324.0 &
787.0 \\

           \cmidrule(r){1-1}
           \cmidrule(rl){2-8}
           \# of Benchmarks   & 72 &58 & 10 & 98 & 50 & 24 & 24   \\
 \bottomrule
\end{tabular}
\end{table}
\label{ap:alg_stats}

\end{document}